\documentclass{article}

\usepackage{esvect}
\usepackage{hhline}

\usepackage{amsfonts,amsmath,amsthm,dsfont}
\usepackage{mathtools}
\usepackage{paralist}
\usepackage{multirow}
\usepackage{comment}
\usepackage{xspace}

\usepackage{colortbl}
\usepackage{tikz}
\usetikzlibrary{decorations,arrows,petri,topaths,backgrounds,shapes,positioning,fit,calc,decorations.pathreplacing,patterns,intersections}

\usepackage[numbers,sectionbib]{natbib}
\usepackage{subcaption}

\usepackage{authblk}

\usepackage{booktabs}
\usepackage{multirow}
\usepackage[textsize=scriptsize]{todonotes}
\usepackage[ruled,linesnumbered]{algorithm2e}

\newtheorem{definition}{Definition}
\newtheorem{example}{Example}

\newtheorem{theorem}{Theorem}
\newtheorem{lemma}{Lemma}
\newtheorem{corollary}{Corollary}

\newlength{\additionaltextheight}
\newlength{\additionaltextwidth}
\newcommand{\comm}[1]{{\color{gray}#1}}

\definecolor{dg}{rgb}{0.3,0.3,0.3}
\definecolor{winered}{rgb}{0.6,0,0.1}
\definecolor{lightblue}{rgb}{0.6,0.88,0.98}

\definecolor{darkgreen}{rgb}{0.0,0.5,0.2}
\definecolor{llgray}{rgb}{0.27, 0.68, 0.68}

\newtheorem{observation}{Observation}

\usepackage[colorlinks,citecolor=green!60!black,linkcolor=red!60!black,pdfdisplaydoctitle,
menucolor=orange!40!black,filecolor=magenta!40!black,urlcolor=blue!40!black,linkcolor=red!40!black,citecolor=green!40!black,colorlinks,pagebackref]{hyperref}

\newcommand{\ppp}{{\cal P}}

\newcommand{\alternativeset}{\ensuremath{A}}
\newcommand{\voterset}{\ensuremath{V}}

\newcommand{\agenda}{\ensuremath{{\cal L}}}
\newcommand{\partialagenda}{\ensuremath{{\cal B}}}

\newcommand{\altsymbol}{\ensuremath{a}}
\newcommand{\votersymbol}{\ensuremath{v}}

\newcommand{\profile}{\ensuremath{\ppp}}

\newcommand{\pref}{\ensuremath{\succ}}
\newcommand{\spref}{{\ensuremath{\scriptstyle\,\succ\,}}}

\newcommand{\profiletuple}{(\alternativeset, \voterset)}

\newcommand{\alternative}{alternative\xspace}
\newcommand{\alternatives}{alternatives\xspace}
\newcommand{\voter}{voter\xspace}
\newcommand{\voters}{voters\xspace}

\newcommand{\Alternative}{Alternative\xspace}

\newcommand{\beat}{beat\xspace}

\newcommand{\agendapref}{{\ensuremath\scriptstyle\,\rhd\,}}

\newenvironment{myquote}{\list{}{\leftmargin=0.2in\rightmargin=0.2in}\item[]}{\endlist}
\newcommand{\probDef}[3]{
  \begin{myquote}
   #1\\
  \textbf{Input:} #2\\
  \textbf{Question:} #3
  \end{myquote}
}

\DeclareMathOperator{\cupdot}{\mathbin{\dot{\cup}}}

\newcommand{\probManipulation}{\textsc{Manipulation}\xspace}
\newcommand{\probWManipulation}{\textsc{Weighted Manipulation}\xspace}
\newcommand{\probAgendaControl}{\textsc{Agenda Control}\xspace}
\newcommand{\probPossibleWinner}{\textsc{Possible Winner}\xspace}
\newcommand{\probNecessaryWinner}{\textsc{Necessary Winner}\xspace}
\newcommand{\probWPossibleWinner}{\textsc{Weigh\-ted Possible Winner}\xspace}
\newcommand{\probWNecessaryWinner}{\textsc{Weigh\-ted Necessary Winner}\xspace}
\newcommand{\probVC}{\textsc{Vertex Cover}\xspace}
\newcommand{\probIS}{\textsc{Independent Set}\xspace}

\newcommand{\NP}{{\mathsf{NP}}}
\newcommand{\coNP}{{\mathsf{coNP}}}
\newcommand{\FPT}{{\mathsf{FPT}}}

\newcommand{\setseq}[1]{\ensuremath{\overrightarrow{#1}}}
\newcommand{\setrevseq}[1]{\ensuremath{\overleftarrow{#1}}}

\newcommand{\thmagendacontrolsuccessive}{
  \probAgendaControl can be solved in $O(n\cdot m^2)$ time for the successive procedure,
  where $n$ denotes the number of voters and $m$ the number of alternatives.
}

\newcommand{\thmagendacontrolamendment}{
  \probAgendaControl can be solved in $O(n\cdot m^2 + m^3)$ time for the amendment procedure,
  where $n$ denotes the number of voters and $m$ the number of alternatives.
}

\newcommand{\thmmanipulationsuccessive}{
  \probManipulation can be solved in $O(n \cdot m)$ time for the successive procedure,
  where $n$ denotes the number of voters and $m$ the number of alternatives.
}

\newcommand{\thmmanipulationamendment}{
  \probManipulation can be solved in $O(n \cdot m^2)$ time for the amendment procedure,
  where $n$ denotes the number of voters and $m$ the number of alternatives.
}

\newcommand{\corwmanipulation}{
  \probWManipulation can be solved in  $O(n \cdot m)$ time for the successive procedure and in $O(n \cdot m^2)$ time for the amendment procedure, 
  where $n$ denotes the number of voters and $m$ the number of alternatives.
}

\newcommand{\thmpossiblewinnersuccessive}{
   \probPossibleWinner with a fixed agenda is $\NP$-complete for the successive procedure.
}

\newcommand{\thmpossiblewinneramendment}{
   \probPossibleWinner with a fixed agenda is $\NP$-complete for the amendment procedure.
}

\newcommand{\thmpossiblewinnerilp}{
  Let $m$ denote the number of alternatives and $n$ the number of voters of a given \probPossibleWinner instance.
  Let $\mathrm{ilp}(\rho_1, \rho_2,  \rho_3)$ denote the running time of the feasibility problem of an integer linear program which has $\rho_1$~variables and $\rho_2$~constraints, and where the maximum of the absolute values of the coefficients and the constant terms is $\rho_3$, $\rho_1, \rho_2, \rho_3\in \mathds{N}$.
  Then,  \probPossibleWinner can be solved in $O(m! \cdot \mathrm{ilp}(m!\mathbin{\cdot}2^{m^2},\, 2^{m^2}+m, n))$ time for the successive procedure and in $O(m! \cdot m^m \cdot \mathrm{ilp}(m!\mathbin{\cdot}2^{m^2},\, 2^{m^2}+3m, 2n))$ time for the amendment procedure. 
}

\newcommand{\corpossiblewinnerfpt}{
   Let $m$ denote the number of alternatives and $n$ the number of voters of a given \probPossibleWinner instance.
   Then, \probPossibleWinner can be solved in $O(\rho^{2.5\rho+o(\rho)+2}\cdot \log{(n+2)})$ time for the successive procedure and the amendment procedure,
   where $\rho = m!\cdot 2^{m^2}$.
}

\newcommand{\thmnecessarywinnersuccessive}{
  \probNecessaryWinner can be solved in $O(n \cdot m^3)$ time for the successive procedure.
}

\newcommand{\thmnecessarywinneramendment}{
  \probNecessaryWinner is $\coNP$-complete for the amendment procedure even with a fixed agenda.
}

\newcommand{\cornecessarywinneramendmentfpt}{
  Let $m$ denote the number of alternatives and $n$ the number of voters of a given \probNecessaryWinner instance.
  Then, \probNecessaryWinner for the amendment procedure can be solved in $O(\rho^{2.5\rho+o(\rho)+2}\cdot \log{(n+2)})$ time,  where $\rho = m!\cdot 2^{m^2}$.
}

\newcommand{\thmweightedpossiblewinner}{
  \probWPossibleWinner for the successive procedure is weakly $\NP$-complete
  even for three alternatives and when the agenda~$\partialagenda$ is linear.
}

\newcommand{\thmweightednecessarywinner}{
  For the successive procedure, \probWNecessaryWinner can be solved in $O(n\cdot m^3)$ time.
  For the amendment procedure, \probWNecessaryWinner can be solved in linear time for up to three \alternatives.
}

\newcommand{\obsmanipulatorsvotethesame}{
 If there is a successful (weighted) manipulation,
 then there is also a successful one where all \voters from the coalition rank the \alternatives in
 the same way.
}

\newcommand{\obstournamentpartition}{
Every tournament can be partitioned into disjoint strongly connected
    strong subtournaments~$T_1 \coloneqq (U_1, A_1), T_2\coloneqq(U_2, A_2), \ldots,
    T_t\coloneqq(U_t, A_t)$, such that 
    \begin{enumerate}
      \item for every two subtournaments~$T_i$ and $T_j$, in the
      original tournament, either all arcs
      are from the vertices of $T_i$ to the vertices of $T_j$ or the
      other way round; and 
      \item the graph resulting from deleting all but one vertex
      from each subtournament~$T_i$, is a directed
      acyclic graph.
    \end{enumerate}
}

\newcommand{\claimone}{
  Let $\alternativeset'\subseteq\alternativeset$ be a subset of
  \alternatives.
  Let $b$ be an \alternative in $\alternativeset \setminus \alternativeset'$ such that
  $b$ ist \emph{not} an $(\alternativeset'\cup \{b\})$-majority winner.
  Then, 
    \begin{enumerate} 
      \item alternative~$b$ is not a successive winner under every agenda~$\agenda$ that extends the partial order~$b \agendapref \alternativeset'$, and 
      \item if such an alternative~$b$ does not exist and if
      $\alternativeset \setminus \alternativeset' \neq \emptyset$, then no \alternative from $\alternativeset'$ can be a successive winner.
    \end{enumerate}
}

\newcommand{\claimtwo}{
  There is a completion~$(\profile^*,\agenda)$ of~$(\profile,\partialagenda)$ 
  satisfying Condition~\ref{cond1} from Observation~\ref{obs:conditions-p-not-necessary-winner} if and only if
  $p$ is not an $(\alternativeset \setminus \Precc{\partialagenda}{p})$-majority winner in the $p$-discriminating completion of $\profile$.
}

\newcommand{\claimthree}{
  Assume that no completion of~$(\profile,\partialagenda)$ satisfies Condition~\ref{cond1}.
  Then, there is a completion~$(\profile^*,\agenda)$ of~$(\profile,\partialagenda)$ 
  satisfying Condition~\ref{cond2} from Observation~\ref{obs:conditions-p-not-necessary-winner} if and only if 
  there is some alternative $b \in \Precc{\partialagenda}{p}$ being 
  a $(\Succc{\partialagenda}{b}\cup \{b\})$-majority winner in the $c$-privileging completion of~$\profile$.
}

\newcommand{\obsstrongtournamenthasHamiltoniancycle}{
 Every strong\-ly connected tournament contains a Hamiltonian cycle.
}

\newcommand{\nth}[1]{$#1^{\text{th}}$\xspace}

\newcommand{\nthmwinner}[1]{\nth{#1}-round manipulated winner}
\newcommand{\nthawinner}[1]{\nth{#1}-round amendment winner}

\title{Parliamentary Voting Procedures:\\
  Agenda Control, Manipulation, and Uncertainty\thanks{A preliminary short 
version of this work has been presented at the 24th International Joint Conference on Artificial Intelligence (IJCAI 2015), Buenos Aires, July, 2015~\citep{BreCheNieWal2015}.
In this long version, we provide two additional fixed-parameter tractability results (Theorem \ref{thm:possiblewinner-ilp} and Corollary~\ref{cor:possiblewinner-fpt}) and all proofs that were omitted in the conference version.}
}

\author[1]{Robert Bredereck}
\author[1]{Jiehua Chen}
\author[1]{Rolf Niedermeier}
\author[2]{Toby Walsh}

\affil[1]{Institut f\"ur Softwaretechnik und Theoretische Informatik,
  TU Berlin, Berlin, Germany\\
  \texttt{\{robert.bredereck, jiehua.chen, rolf.niedermeier\}@tu-berlin.de}}
\affil[2]{NICTA and the University of New South Wales, Australia,\hspace{4cm}
  \texttt{toby.walsh@nicta.com.au}}
\date{}

\begin{document}

\maketitle

\thispagestyle{plain}
\setcounter{footnote}{0}

\begin{abstract}
We study computational problems for two popular parliamentary voting procedures: the
amendment procedure and the successive procedure.
While finding successful manipulations or agenda controls is tractable
for both procedures,
our real-world experimental results indicate that most
elections cannot be manipulated by a few voters and agenda
control is typically impossible.
If the voter preferences are incomplete, then finding which alternatives can possibly win is $\NP$-hard for both procedures. Whilst deciding if an alternative necessarily wins is $\coNP$-hard for the amendment procedure,
it is polynomial-time solvable for the successive one.

\end{abstract}

\section{Introduction}
\label{sec:introduction}

Two prominent voting rules 
are used in many parliamentary chambers 
to amend and decide upon new legislation: 
the successive procedure and the
amendment 
procedure~\citep{ApeBalMas2014}. 
Both are \emph{sequential voting procedures}:
the alternatives are ordered (thus forming an agenda) and they are 
considered one by one, making a binary decision based on
majority voting in each step.
In a nutshell, in each step, the \emph{successive procedure} considers the current alternative and decides whether to accept it (in which case the procedure stops and the winner is determined)
or to reject it and the procedure then continues with the remaining alternatives in the given order.
The \emph{amendment procedure} in each step 
jointly considers two current alternatives and decides by majority 
voting which one of the two is eliminated---the remaining one then 
will be confronted with the next alternative on the agenda.

There are many reasons to study the properties of parliamentary voting procedures, and especially
to consider computational questions. 
First, parliamentary voting procedures are
used very frequently in practice. 
For example, the recent 112th Congress of the US Senate
and House of Representatives had 1030 votes
to amend and approve 
bills. This does not take into account the hundreds
of committees that also amended and voted on these bills.
As a second example, there were 351 divisions within
the Houses of Lords and Commons in 2013 to amend or approve bills. 

Second, parliamentary voting procedures are
used to make some of the most important decisions in society. 
We decide to reduce carbon emissions, provide universal health care,
or raise taxes based on the outcome of such voting 
procedures. When rallying support for new
legislation, it is vital to know what amendments can and cannot be passed. 
Fortunately, we have excellent historical voting records for
parliamentary chambers. 
We can therefore make high quality 
predictions about how sincere or ``sophisticated'' voters will vote.
 
Third, \citet{EneKoe1980} give evidence that parliamentary voting may
be strategic: 
\begin{quote}
{\em ``Thus, it is shown that sophisticated voting 
does occur in Congress and in fact is encouraged by
the way amendments are used in the legislative process.
It should not come as a surprise to congressional
scholars that congressmen do not always vote sincerely.''
}
\end{quote}
Fourth, there is both theoretical
and empirical evidence that the final outcome
critically depends on the order in which amendments
are presented. For example, \citet{OrdSch1987}
remark that 
\begin{quote}
  {\em ``\ldots legislative decisions are at the mercy of elites who control agendas.''
}
\end{quote}

It is therefore interesting to
ask if, for example, computational complexity
is a barrier to the control of the agenda
or to strategic voting in such parliamentary
voting procedures. 
The former refers to the \textsc{Agenda Control} problem while
the latter to the \textsc{Manipulation} problem.
It is also interesting to ask 
if we can efficiently compute whether
a particular amendment can possibly (or necessarily) pass despite uncertainty in
the votes or the agenda. 
This refers to the \textsc{Possible Winner} (or \textsc{Necessary Winner}) 
problem.
We provide one of the first computational studies of these issues,
giving both theoretical and empirical results.

\subsection{Related Work}
There are many studies in the economic and political literature on parliamentary voting procedures,
starting with \citet{Black1958} and \citet{Far1969}, concerning ``insincere'' or 
``sophisticated'' or ``strategic''
voting e.g.~\citep{Mil1977,EneKoe1980,MN81,SW84,Ban1985,Mou86,OrdSch1987}. 
\citet{ApeBalMas2014} characterize both the amendment and
the successive procedures from an axiomatic perspective.

\citet{Mil1977} studies the set of \alternatives that may win, 
the \textsc{Agenda Control} problem.
In particular, he
shows 
that under sincere voting, an \alternative can become an amendment winner under some agenda
if and only if it
belongs to the Condorcet set (also known as top cycle).
We extend this result by a constructive proof. 
For the successive procedure, however, he only 
shows that every \alternative from the Condorcet set can win.
\citet{BarGer2014} follow Miller's research of characterizing the set of \alternatives that may become an amendment (or a successive) winner by controlling the agenda.\footnote{Their definitions for both procedures are actually slightly different from ours, the common ones.
}
\citet{
  Ras2014} empirically examines the behavior of voters in
the Norwegian parliament, where the successive procedure is used.
He reports that successful insincere voting, where voters may vote
differently from their true preferences and the outcome is better
for them, is very rare.

Using computational complexity as a barrier against \textsc{Manipulation} (that is, changing the outcome by adding voters) was initiated by \citet{BarTovTri1989}. 
They show that manipulating a special variant of the Copeland voting rule is $\NP$-hard. 
\citet{BarOrl1991} show that manipulating the Single Transfer Vote~(STV) voting rule is $\NP$-hard even the number of voters allowed to add is one.
This voting rule is used in the parliamentary elections of many countries. 
It is a sequential voting procedure and
works similarly 
to the successive procedure except that there is no agenda. 
Instead, in each step, the \alternative that is ranked first by the least number of \voters will be deleted from the profile.
The $\NP$-hardness result for manipulating STV is of particular interest
since we show polynomial-time results for manipulating the successive procedure.
These two complexity results indicate that it is the agenda that makes an important difference.

Concerning uncertainty in elections (that is, each voter's preference order may be not a linear order),
there is some work in the 
political literature~\citep{OrdPal1988,Jun1989},
but there seems to be significantly 
more activity on the computational side.
\citet{KonLan2005} consider the \textsc{Possible Winner} and \textsc{Necessary 
Winner} problems for the
Condorcet rule.
The same problems for several other common voting rules have been frequently studied
~\citep{Wal2007,BetHemNie2009,
  HazAumKraWoo2012,AziHarBriLanSee2012}.

\citet{Mou86} discusses a generalization of the amendment procedure: the voting tree procedure.
The amendment procedure is a special case of the voting tree procedure~\citep{Mou86}.
This general procedure
employs a binary voting tree where the leaves represent the alternatives and each \alternative is represented by at least one leaf, 
and each internal node represents the alternative that wins the pairwise comparison of its direct children.
The alternative represented by the root defines the winner.
If the binary tree is degenerate and if each alternative is represented by
exactly one leaf, then this procedure is identical to the amendment procedure.
To tackle the \textsc{Manipulation} problem with weighted voters, 
\citet{ConSanLanjacm2007} provide a cubic-time algorithm for the 
voting tree procedure 
while our quadratic-time algorithm is tailored for the amendment procedure.
\citet{XiaCon2011} provide intractability results for the \textsc{Possible} (resp.\ \textsc{Necessary}) \textsc{Winner} problem with weighted voters when the given tree is balanced.
\citet{PinRosVenWal2011} and \citet{LanPinRosSalVenWal2012} show that the \textsc{Possible}
(resp.\ \textsc{Necessary}) \textsc{Winner} problem with weighted voters is intractable for a constant number of \voters (see Table~\ref{tab:result}). 



\newcommand{\mythm}[1]{Thm. \ref{#1}}
\newcommand{\millerfollow}{$^{\heartsuit}$ }
\newcommand{\langpinifollow}{$^\spadesuit$ }
\newcommand{\thmcite}[1]{Thm.~\ref{#1}}
\newcommand{\thmscite}[2]{Thm.~\ref{#1} and \ref{#2}}
\newcommand{\thmcorcite}[2]{Thm.~\ref{#1} and Cor.~\ref{#2}}
\newcommand{\corcite}[1]{Cor.~\ref{#1}}
\begin{table}[t]
  \centering
  \begin{tabular}[t]{ @{}l l l l @{}}
    \toprule
    Problem & Successive & Amendment & References\\
    \midrule
    \probAgendaControl & $O(n\cdot m^2)$  & $O(n\cdot m^2 + m^3)$\millerfollow & \thmscite{thm:control-successive-p}{thm:control-amendment-p}\\ 
    & & \\[-0.7em]
    W. \probManipulation&  $O(n \cdot m)$   &  $O(n \cdot m^2)$ & \thmscite{thm:manipulation-successive-p}{thm:manipulation-amendment-p}, \corcite{cor:wmanipulation-p}\\ 
    & & \\[-0.7em]
    \probPossibleWinner & $\NP$-c & $\NP$-c  & \thmscite{thm:possiblewinner-successive-nph}{thm:possiblewinner-amendment-nph}\\ 
    & & \\[-.9em]
    & $\FPT$ wrt. $m$ & $\FPT$ wrt. $m$ & \corcite{cor:possiblewinner-fpt}\\
    & & \\[-0.7em]
    \probNecessaryWinner & $O(n \cdot m^3)$ & $\coNP$-c & \thmscite{thm:necessarywinner-successive-p}{thm:necessarywinner-amendment-nph}\\ 
    & & \\[-.9em]
    &  & $\FPT$ wrt. $m$ & \corcite{cor:necessarywinner-amendment-fpt}\\
    & & \\[-0.7em]
    W.~\probPossibleWinner & $\NP$-c~($m=3$) & $\NP$-c~($m=3$)\langpinifollow & \thmcite{thm:wpossiblewinner-nph}\\ 
    & & \\[-0.7em]
    W.~\probNecessaryWinner & $O(n \cdot m^3)$ & $O(n)$ for $m \le 3$ & \thmcite{thm:w_necessary} \\
    & & \\[-.9em]
    &  & $\coNP$-c~($m=4$)\langpinifollow  & ---\\
    \bottomrule
  \end{tabular}
  \caption{Computational complexity of the considered problems. 
    ``W.'' refers to the relevant problem with weighted voters.
    ``$\NP$-c'' (resp.\ ``$\coNP$-c'') stands for $\NP$-complete (resp.\ $\coNP$-complete).
    The number of \voters is denoted by~$n$, the number of \alternatives is denoted by~$m$, and the sum of weights of all voters in the manipulation is denoted by~$w$.
    ``$\FPT$ wrt. $m$'' stands for 
    ``fixed-parameter tractable with respect to $m$'' and means that 
    if the number~$m$ of alternatives is a constant,
    then the relevant problem is polynomial-time solvable and the degreee of the polynomial does not depend on $m$.
    The result marked with
    \millerfollow also follows from the work of \protect\citet{Mil1977}.
    Those marked with \langpinifollow follow from the work of
    \protect\citet{PinRosVenWal2011} and \protect\citet{LanPinRosSalVenWal2012}. 
    Entries containing statements of the form ``$\NP$-c~($z$)''
    (resp.\ ``$\coNP$-c~($m=z$)'') mean that the relevant problem is
    in $\NP$ (resp.\ $\coNP$) and is $\NP$-hard (resp.\ $\coNP$-hard) even with only~$z$~\alternatives.
    All hardness results already hold when the agenda is a linear order.
  }
  \label{tab:result}
\end{table}

\subsection{Our Contributions}

We investigate computational problems for two prominent parliamentary voting procedures:
the successive procedure and the amendment procedure.
We study three types of voting problems.
First, we study whether there is an agenda under which a given alternative can win when the voters vote sincerely; we call the corresponding problem \textsc{Agenda Control}.
Second, we study whether a given alternative can win by adding a given amount of voters; we call the corresponding problem \textsc{Manipulation}.
Third, we study whether a given alternative can possibly (resp.\ necessarily) win when the voters may have incomplete preferences; we call the corresponding problem \textsc{Possible} (resp.\ \textsc{Necessary}) \textsc{Winner}.
See Table~\ref{tab:result} for an overview of our theoretical results.

Our polynomial-time results for the agenda control problem and the manipulation problem indicate that the amendment procedure is computationally more expensive than the successive procedure.
From a computational perspective,
this implies that the amendment procedure may be more resistant to strategic voting and agenda control than the successive procedure.
If voters' preference orders are incomplete,
then deciding whether an \alternative possibly wins is $\NP$-complete for both procedures
while deciding whether an \alternative necessarily wins is polynomial-time solvable for the successive procedure,
but is $\coNP$-complete for the amendment procedure.

Our experiments on \textsc{Agenda Control} and on \textsc{Manipulation} using real-world voting data indicate that 
while both problems are polynomial-time solvable, 
a successful agenda control is very rare and 
a successful manipulation on average needs a coalition containing more 
than half of the voters.

\subsection{Organization of the paper}

In Section~\ref{sec:definitions}, we provide definitions regarding voter preferences, 
preference profiles, and our central parliamentary voting procedures.
In Section~\ref{sec:agenda-control}, 
we focus on the \textsc{Agenda Control} problem when 
voters are voting sincerely and
we study the corresponding computationally complexity for 
both parliamentary procedures.
In Section~\ref{sec:manipulation},
we deal with strategic behavior of voters and investigate the computational complexity of coalitional manipulations.
In Section~\ref{sec:possible-necessary}, 
we are interested in situations with uncertainty, 
that is, when voters' preferences and the agenda are still incomplete.
We study the problem of whether an \alternative can possibly/necessarily win in such a situation.
In Section~\ref{sec:experiments}, we complement our theoretical study with an experimental evaluation of the \textsc{Agenda Control} problem and the \textsc{Manipulation} problem for real-world profiles.
Section~\ref{sec:conclusion} concludes our work and presents some challenges for future research. 

\section{Preliminaries}
\label{sec:definitions}

Let $\alternativeset\coloneqq\{\altsymbol_1,\ldots,\altsymbol_m\}$ be a set of $m$~\alternatives
and let $\voterset\coloneqq\{\votersymbol_1,\ldots,\votersymbol_n\}$ be a set of $n$~voters. 
A \emph{preference profile}~$\profile\coloneqq\profiletuple$ specifies the
\emph{preference order}s of the voters in $\voterset$,
where each voter~$\votersymbol_i$ ranks the \alternatives according
to the partial order~$\pref_i$ over $\alternativeset$.
For two \alternatives~$b, c \in \alternativeset$, 
the relation~$b\pref_i c$ means that voter~$\votersymbol_i$ strictly prefers $b$ to $c$.
Given two disjoint subsets~$B, C\subseteq \alternativeset$ of alternatives, 
we write $B \pref_i C$
to express that voter~$v_i$ prefers set~$B$ to set~$C$,
that is, 
for each alternative $b\in B$ and each alternative~$c\in C$ it holds that~$b\pref_i c$,
and all alternatives in $B$ (resp.\ $C$) are \emph{incomparable} to each other.
We write $B\pref_i c$ as shorthand for $B\pref_i \{c\}$ and $c\pref_i B$ for $\{c\}\pref_i B$.


We say that alternative~$b$ \emph{\beat{s}} alternative~$c$ (in a head-to-head contest) 
when a majority of voters prefers $b$ to $c$,
and call $b$ the \emph{survivor} and $c$ the \emph{loser} of the two
alternatives.
We call $b$ a \emph{Condorcet winner} if it beats every other \alternative.

We can use a directed graph to illustrate the comparisons between every two \alternatives. 
\begin{definition}[(Weighted) majority graph]
\label{def:majority graph}
  Given a preference profile~$\profile$ $=\profiletuple$,
  we construct an arc-weighted directed graph~$G \coloneqq (U, E)$,
  where $U$ consists of a vertex~$u_j$ for each \alternative~$c_j\in \alternativeset$ 
  and where there is an arc from vertex~$u_{j}$ to vertex~$u_{j'}$ with weight~$w$ 
  if exactly $w$~voters prefer $c_{j}$ to~$c_{j'}$, that is, $c_{j} \pref c_{j'}$.
  We call the constructed graph~$G$ a \emph{weighted majority graph} for profile~$\profile$.
  
  If we ignore the weights and the arcs with weights less than $|\voterset|/2$,
  then we obtain a \emph{majority graph} (without weights) for the profile~$\profile$. 
  We call the constructed majority graph a \emph{tournament} if there is exactly one arc between every two distinct vertices. 
\end{definition}

We give a small example to illustrate the concept of (weighted) majority graphs.

\begin{example}
  \label{ex:profile}
  Let $\profile$ be a preference profile with three \alternatives~$a, b, c$,
  and three voters~$v_1, v_2, v_3$ whose preference orders are specified as follows:
  \begin{align*}
    v_1\colon a \pref b \pref c,\quad  v_2\colon b \pref a \pref c, \quad
    v_3\colon c \pref a \pref b.
  \end{align*}
  
  The weighted majority graph for $\profile$ consists of three alternatives and six weighted arcs as depicted in the left figure below.
  Bold arcs indicate the majority relation.
  
  {\centering
    \tikzstyle{majarr}=[draw=black,thick,->]
    \tikzstyle{alter}=[circle, minimum size=18pt, inner sep = 2.5pt, draw=black!15, ultra thick, fill=black!9]
    \tikzstyle{innercircle}=[circle, minimum size=16pt, ultra thick, draw=white]
    
    \newcommand{\sss}{\ensuremath}
 
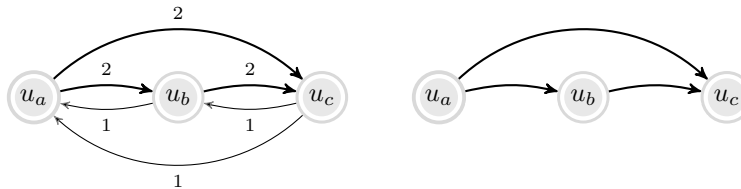
  \begin{tabular}{ccc}
  \begin{tikzpicture}[auto, >=stealth']
    \tikzstyle{minarr}=[draw=black!50,->]
    \node[alter] at (0,0) (a) {$u_{\sss a}$};
    \node[alter, right = 8ex of a] (b) {$u_{\sss b}$};
    \node[alter, right = 8ex of b] (c) {$u_{\sss c}$};

    \node[innercircle] at (a) {};
    \node[innercircle] at (b) {};
    \node[innercircle] at (c) {};

    \draw[majarr] (a) edge[bend left=13] node[midway, anchor=south] {$\scriptstyle 2$} (b);
    \draw[majarr] (b) edge[bend left=13] node[midway, anchor=south] {$\scriptstyle 2$} (c);
    \draw[majarr] (a) edge[bend left=43] node[midway, anchor=south] {$\scriptstyle 2$} (c);
    \draw[minarr] (b) edge[bend left=13] node[midway, anchor=north] {$\scriptstyle 1$} (a);
    \draw[minarr] (c) edge[bend left=13] node[midway, anchor=north] {$\scriptstyle 1$} (b);
    \draw[minarr] (c) edge[bend left=43] node[midway, anchor=north] {$\scriptstyle 1$} (a);
  \end{tikzpicture}
  & \quad &
  \begin{tikzpicture}[auto, >=stealth']
    \tikzstyle{minarr}=[white]
    \node[alter] at (0,0) (a) {$u_{\sss a}$};
    \node[alter, right = 8ex of a] (b) {$u_{\sss b}$};
    \node[alter, right = 8ex of b] (c) {$u_{\sss c}$};

    \node[innercircle] at (a) {};
    \node[innercircle] at (b) {};
    \node[innercircle] at (c) {};

    \draw[majarr] (a) edge[bend left=13] (b);
    \draw[majarr] (b) edge[bend left=13] (c);
    \draw[majarr] (a) edge[bend left=43] (c);
    \draw[minarr] (b) edge[bend left=13] node[midway, anchor=north] {$\color{white}\scriptstyle 1$} (a);
    \draw[minarr] (c) edge[bend left=13] node[midway, anchor=north] {$\color{white}\scriptstyle 1$} (b);
    \draw[minarr] (c) edge[bend left=43] node[midway, anchor=north] {$\color{white}\scriptstyle 1$} (a);
  \end{tikzpicture}
  \end{tabular}
  \par}
  The corresponding majority graph without weights is depicted in the right figure.
  We can verify that it is a tournament.
  Vertex~$u_a$ has exactly two out-going arcs, meaning that alternative~$a$ beats every other alternative. 
  Thus, it is a Condorcet winner.
\end{example}

Given two preference profiles~$\profile$ and $\profile'$ for the same set of \alternatives and the same set of \voters, 
we say that $\profile$ \emph{extends} $\profile'$ if for every~$i$, preference order~$\pref_i$ from~$\profile$ includes the preference order~$\pref'_i$ from~$\profile'$.
If each~$\pref_i$ is already a linear order, 
then we also say that $\profile$ \emph{completes} $\profile'$.

Consider a preference profile~$\profile\coloneqq \profiletuple$.
Let $\alternativeset'\subseteq \alternativeset$ be a subset of \alternatives.
If not specified explicitly, let $\setseq{A'}$ denote an arbitrary but fixed linear order of the
alternatives in~$A'$. Consequently, $\setrevseq{A'}$ denotes the 
corresponding 
reversed order of the alternatives in~$A'$.
Given an alternative~$a\in \alternativeset'$ from
this set, we say that $a$ is an \emph{$\alternativeset'$-majority winner} if 
a majority of \voters prefer $a$ to all alternatives from $\alternativeset'\setminus \{a\}$, that is,
$|\{v_i \mid a \pref_i A'\setminus \{a\}\}| > |V|/2$.



We consider two of the most common parliamentary voting procedures.
For both procedures, we assume that a linear order over the $m$~\alternatives in $\alternativeset$ is given.
We refer to this linear order~$\agenda$ as an \emph{agenda}.
If this order is \emph{not} linear,
then we call it a \emph{partial agenda}, denoted by~$\partialagenda$.
We use the symbol~$\agendapref$ for the agenda order to distinguish it from the preference order~$\pref$. 

\begin{definition}[Successive procedure]\label{def:successive}
There are at most $m$ rounds with $m$ being the number of alternatives. 
Starting with round~$i\coloneqq1$, we repeat the following until we make a decision:
Let $c$ be the $i$th \alternative in the given agenda~$\agenda$.
If a majority of voters prefers alternative~$c$ 
to all \alternatives that are ordered behind it in $\agenda$,
then $c$ is the decision and we call it a successive winner.
Otherwise, we 
proceed to round~$i \coloneqq i+1$.
\end{definition}

In Europe, the successive procedure is used in many parliamentary
chambers including those of Austria, Belgium, Denmark,
France, Germany, Greece, Iceland, Ireland, Italy,
Luxembourg, the Netherlands, Norway, Portugal, 
and Spain~\citep{Ras2000}.
We give a small example to illustrate how the successive procedure works.
\begin{example}\label{ex:successive}
Consider the profile from Example~\ref{ex:profile} and consider the following agenda~$\agenda\colon a \agendapref b \agendapref c$.
  Since less than half of the voters prefer $a$ to $\{b, c\}$ (only $v_1$ does),
  $a$ is not a successive winner.
  Since more than half of the voters prefer $b$ to $c$ (voters~$v_1$ and $v_2$), $b$ is the successive winner.
\end{example}

\begin{definition}[Amendment procedure]\label{def:amendment}
This procedure has $m$~rounds with $m$ being the number of alternatives.
In the $1^{\text{st}}$ round, we let the $1^{\text{st}}$-round amendment
winner be the first alternative in the given agenda~$\agenda$.
Then, for each round~$2\le i \le m$, 
let the \nthawinner{i} be the survivor between 
the $i$th alternative in $\agenda$ and the \nthawinner{(i-1)}.
We define the \nthawinner{m} to be the amendment winner.
\end{definition}

In Europe, the amendment procedure is used in the parliamentary
chambers of Finland, Sweden, Switzerland, and the United
Kingdom. It is also used in the U.S. Congress and several
other countries with Anglo-American ties~\citep{Ras2000}.
We use the same example given for the successive procedure to illustrate how the amendment procedure works.
It may be helpful to consider the corresponding majority graph (see Definition~\ref{def:majority graph}).

\begin{example}\label{ex:amendment}
Consider the profile and the agenda given in Example~\ref{ex:successive} for the successive procedure.
Alternative~$a$ is the $1^{\text{st}}$-round winner since it is the first \alternative on the agenda. 
Since a majority of voters prefers $a$ to $b$, 
$a$ is also the $2^{\text{nd}}$-round winner.
Since a majority of voters also prefers $a$ to $c$, 
$a$ is the $3^{\text{rd}}$-round and hence, the amendment winner.
Indeed, as one can see from the majority graph shown in Example~\ref{ex:profile}, alternative $a$ is a Condorcet winner.
Consequently, it is always an amendment winner no matter how the agenda looks like. 
\end{example}


For the remainder of this paper, we assume that the number of voters is \emph{odd} to
reduce the impact of ties, and break remaining ties.
Furthermore, we consider both unweighted voters and voters with integer weights. 
The weighted case is especially interesting in the parliamentary setting: 
First, there are parliamentary chambers
where voters are weighted (for instance, in the Council of Europe,
preference orders are weighted by the size of the country). 
Second, voters will often vote along party lines. This effectively gives us 
parties casting weighted preference orders. 
Third, the weighted case can inform the situation where
we have uncertainty about the preference orders. 
For example, Theorem 15 of \citet{ConSanLanjacm2007} proves that if
the manipulation problem for a voting rule
is $\NP$-hard for weighted voters who have complete preference orders,
then deciding who possibly wins in the unweighted case is
$\NP$-hard even when there is only a limited form of uncertainty about the preference orders. It would be interesting to prove similar results
about uncertainty and weighted preference orders in parliamentary voting procedures. 

\section{Agenda Control}
\label{sec:agenda-control}

The order of the \alternatives, that is, the agenda, may depend
on the speaker, the Government, logical
considerations (for instance, the status quo goes last, the
most extreme \alternative comes first), the chronological
order of submission, or other factors.
The agenda used can have a major impact on the final decision. It is worth noting
also that there are many possible agendas. 
For example, suppose \voters are sincere and we use 
the successive procedure. 
Then the Condorcet winner is only guaranteed to win 
if it is introduced
in one of the last two positions in the agenda. 
We therefore consider the following computational question for the situation where voters vote sincerely.

\probDef{\probAgendaControl}{
  A preference profile~$\profile\coloneqq\profiletuple$ with linear preference
  orders 
  and a preferred \alternative~$p\in \alternativeset$.
}
{
  Is there an agenda for $\alternativeset$ such that $p$ is the overall winner?
}

We find that both voting procedures are ``vulnerable'' to agenda controls. 
In particular, 
we show how to find in polynomial time an appropriate agenda (if it exists) such that the preferred \alternative can become a successive (resp.\ amendment) winner.  
In the remainder of this section, 
we assume that a preference profile $\profile\coloneqq\profiletuple$ with linear preference
orders and a preferred \alternative~$p\in \alternativeset$ are given.
We denote by $n$ the number of \voters and by $m$ the number of \alternatives in $\profile$.

\subsection{Successive procedure}
The basic approach to controlling the successive procedure is to build an agenda
from back to front such that each of the \alternatives that are currently
among the highest positions in the partial agenda may be strong enough to
beat~$p$ alone but is too weak to be a majority winner against the whole set of \alternatives behind it. 
To formalize this idea, we need some observations.
\begin{observation}\label{obs:abeatsb}
  Let $b, c$ be two \alternatives such that $b$ beats $c$.
  Then, as long as $b$ is behind $c$ in the agenda,
  $c$ can never be a successive winner.
\end{observation}

\begin{proof}
  Let $\agenda$ be an agenda where $b$ is positioned behind $c$.
  Suppose that no alternative in front of $c$ in agenda~$\agenda$ 
  is a successive winner. 
  In the round where $c$ is considered, 
  $c$ is \emph{not} ranked first by a majority
  of \voters because $b$ beats $c$ and $b$ is not yet deleted,
  implying that $c$ will be deleted and is not a successive winner.
\end{proof}

We can generalize the above observation to a subset of \alternatives.

\begin{lemma}\label{lem:claimone}
  \claimone
\end{lemma}

\begin{proof}
  The first statement is easy to see from the definition of a successive winner.
  For the second statement, assume that 
  $\alternativeset \setminus \alternativeset' \neq \emptyset$ and that
  \begin{align*}
    \text{every \alternative }~ c \in \alternativeset\setminus \alternativeset' \text{ is an } (\alternativeset' \cup \{c\})\text{-majority winner.} \tag{$\star$} \label{no majority loser}
    \end{align*}
  Suppose for the sake of contradiction that there is an agenda~$\agenda$ under which the successive winner~$a'$ would come from $\alternativeset'$.
  This implies that $a'$ beats every alternative behind it in $\agenda$;
  by assumption~(\ref{no majority loser}),
  $a'$ is beaten by every alternative in $\alternativeset\setminus \alternativeset'$.
  This further implies that 
  $\agenda\colon (\alternativeset \setminus \alternativeset') \agendapref a'\text{.}$      
  Let $d$ be the last alternative in $\alternativeset \setminus \alternativeset'$ that is still in front of $a'$, that is,
  \begin{align*}
    \agenda \text{ satisfies } ( (\alternativeset \setminus \alternativeset') \setminus \{d\} ) \agendapref d \agendapref a'\text{.}
  \end{align*}
  By assumption~(\ref{no majority loser}), $d$ would become a successive winner if all alternatives in front of it in agenda~$\agenda$ are deleted (note that all alternatives behind $d$ come from $\alternativeset'$)---a contradiction to $a'$ being a successive winner.
\end{proof}

By the above lemma, we can construct an agenda from back to front by first placing our preferred \alternative~$p$ at the last
position and setting $\alternativeset'\coloneqq\{p\}$.
We will extend the agenda by putting those alternatives~$c$ that are not $(\alternativeset'\cup \{c\})$-majority winners right in front of $\alternativeset'$.
Using this approach we can solve the agenda control problem for the successive procedure in polynomial time.

\begin{theorem}\label{thm:control-successive-p}
 \thmagendacontrolsuccessive
\end{theorem}

\begin{proof}
  We first describe the algorithm and then analyze its correctness and its running time.
  
  Let $\alternativeset' \coloneqq \{p\}$ and let the agenda $\agenda$ consist of only one \alternative, namely~$p$.
  Execute the following steps.
  \begin{enumerate}[(1)]
    \item\label{stop} If $\alternativeset' = \alternativeset$, then we have constructed an agenda~$\agenda$ under which $p$ can win and we answer ``yes''.
    \item\label{everyone-is-majority} If $\alternativeset \setminus \alternativeset' \neq \emptyset$ 
    and if every alternative~$c\in \alternativeset \setminus \alternativeset'$ is an $\alternativeset' \cup \{c\}$-majority winner,
    then we stop and answer ``no''.
    \item\label{not-everyone-is-majority} Otherwise, for each alternative~$c\in \alternativeset \setminus \alternativeset'$ that is \emph{not} an $\alternativeset' \cup \{c\}$-majority winner,
    we extend the agenda~$\agenda$ by putting alternative~$c$ right in front of all alternatives in $\alternativeset'$ and 
    we let $\alternativeset' \coloneqq \alternativeset' \cup \{c\}$.
    Proceed with Step~(\ref{stop}).
  \end{enumerate}
  
  For the correctness, if Step~(\ref{everyone-is-majority}) ever applies, 
  then by the second statement in Lemma~\ref{lem:claimone}, no alternative in $\alternativeset'\ni p$ can win, implying that $p$ can never win.
  Thus, we can safely reply with ``no''.
  
  If Step~(\ref{everyone-is-majority}) never applies, we show that if $p$ is a successive winner of a profile restricted to the alternative set~$A'$,
  then it is also a successive winner of the profile that additionally contains a non-$(A'\cup \{c\})$-majority winner:
  let $c$ be an alternative in $\alternativeset \setminus \alternativeset'$ that is \emph{not} an $(\alternativeset' \cup \{c\})$-majority winner.
  Assume that $p$ is a successive winner under the current agenda~$\agenda$ for the profile~$\profile'$ restricted to the alternatives in $\alternativeset'$.
  Then, by the first statement in Lemma~\ref{lem:claimone}, 
  it follows that under every agenda that extends $c \agendapref \alternativeset'$,
  $c$ is not a successive winner. 
  This means that the procedure would delete $c$ and go on with the alternatives in $\alternativeset'$.
  By assumption, $p$ is a successive winner for profile $\profile'$ and agenda~$\agenda$.
  Therefore, in the profile restricted to the alternatives in $\alternativeset'\cup \{c\}$, 
  the agenda~$p \agendapref \agenda$ also makes $p$ win.

  We now come to the running time analysis.
  First, if Steps~(\ref{stop}) and (\ref{not-everyone-is-majority}) do not apply, 
  then Step~(\ref{everyone-is-majority}) applies which leads to termination.
  Second, in Step~(\ref{not-everyone-is-majority}),
  for every alternative $c\in \alternativeset\setminus \alternativeset'$,
  we check whether it is an $\alternativeset \cup \{c\}$-majority winner.
  This check can be done in $O(n)$~time:
  We maintain a list~$T$ of size $n$ that stores for each voter~$v$,
  the highest position of \alternative from $\alternativeset'$ ranked by $v$.
  We iterate over every voter~$v$ and compare the position~$v(c)$ of $c$ ranked by $v$ and the position~$T(v)$ stored by $T$ for voter~$v$.
  We count the number of times where $v(c) < T(v)$.
  If this number is smaller than $n/2$,
  then $c$ is not an $\alternativeset'\cup\{c\}$-majority winner; we add $c$ to $\alternativeset'$ 
  and we update the list~$T$ by changing the entry~$T(v)$ where $v(c)< T(v)$.

  Since we may execute Step~(\ref{not-everyone-is-majority}) at most $m$ times (at most $m$ alternatives may be added to $\alternativeset'$),
  the total running time is $O(n\cdot m^2)$.
\end{proof}

\subsection{Amendment procedure}
Controlling the amendment procedure is closely related to finding a Hamiltonian cycle in a strongly connected tournament. To see this, we first construct a \emph{majority graph} for the given preference profile (see the corresponding definition in Section~\ref{sec:definitions}).
Recall that we assume the number of \voters to be odd. 
The majority graph has $m$~vertices and ${m\choose 2}$~arcs and is indeed a tournament.
{F}rom the theory of directed graphs~\cite[Thm. 7]{HarMos1966},
we can conclude that every strongly connected tournament contains a Hamiltonian cycle.
Now, the crucial idea is to check whether the vertex that corresponds to $p$
belongs to a strongly connected component that has only out-going arcs.
\Alternative~$p$ can win under an appropriate agenda if and only if this is the case.
  
\begin{observation}[Theorem 7, \citep{HarMos1966}]
  \label{obs:strong_tournament->Hamiltonian_cycle}
  \obsstrongtournamenthasHamiltoniancycle
\end{observation}

\noindent By carefully examining the constructive proof for Observation~\ref{obs:strong_tournament->Hamiltonian_cycle}, 
we can find a Hamiltonian cycle in $O(m+m^2 + m \cdot (m+m^2)) = O(m^3)$ time.

Utilizing the fact that the underlying undirected graph for a tournament is complete, 
we obtain the following:

\begin{observation}\label{lem:tournament-is-strong-components-dag}
  \obstournamentpartition
\end{observation}

From this observation, we can derive the next theorem.
Note that \citet{Mil1977} already characterized the set of alternatives that can become an amendment winner under an appropriate agenda.
Our theorem strengthens this result by giving a polynomial-time algorithm.

\begin{theorem}\label{thm:control-amendment-p}
  \thmagendacontrolamendment
\end{theorem}

\begin{proof}
  By Observation~\ref{lem:tournament-is-strong-components-dag},
  every tournament consists of strongly connected subtournaments which can be ordered by topological sorting.
  Now, observe that only the \alternatives corresponding to the
  vertices from the top-most strongly connected subtournament can become an amendment winner.
  In other words, if the vertex corresponding to the preferred
  \alternative~$p$ does not belong to the top-most subtournament, then $p$ can never win. 
  By carefully examining the constructive proof for Observation~\ref{obs:strong_tournament->Hamiltonian_cycle}, 
  we can find a Hamiltonian cycle for a strongly connected tournament in $O(m^3)$ time.  
  Now, we construct a sequence~$L_{\textrm{ver}}$ of vertices by reversing the
  orientation of the Hamiltonian cycle, starting with the predecessor of the vertex~$u_p$ corresponding to $p$
  and ending at the vertex~$u_p$, and let $L_{\textrm{alt}}$ be the order of the
  alternatives corresponding to $L_{\textrm{ver}}$. 
  We can verify that $p$ is an amendment winner for every agenda that extends order~$L_{\textrm{alt}}$.

  Thus, the problem is reduced to finding strongly connected
  subtournaments of the majority graph for the given preference
  profile:
  If the vertex corresponding to $p$ is in the top-most subtournament,
  then construct an arbitrary but fixed order that extends
  $L_{\textrm{alt}}$ and we answer ``yes''.
  Otherwise we answer ``no''.

  Now, we come to the running time.
  Constructing a majority graph for a profile takes $O(n\cdot m^2)$ time.
  Partitioning the majority graph into strongly connected components
  takes $O(m+m^2)$ time and checking whether the vertex corresponding to
  $p$ belongs to the top-most component takes $O(m)$ time.
  Finally, finding a Hamiltonian cycle in a strongly connected
  tournament takes $O(m^3)$ time.
  Thus, in total, the approach to controlling the amendment procedure
  takes $O(n\cdot m^2 + m^3)$ time.
\end{proof}

We close this section with two remarks.
First, the approach for the successive procedure actually works for both odd as well as even number of voters.
Second, our approach for the amendment procedure
can be extended to the case where the number of voters is even. 
There, \alternative~$p$ is a winner if and only if no strongly connected component ``dominates'' the strongly connected component that contains $u_p$.

\section{Manipulation}
\label{sec:manipulation}

In this section, we consider the question of how difficult it is computationally 
for \voters to vote strategically to ensure
a given outcome supposing that the other \voters vote sincerely.

\probDef{\probManipulation}
{A profile~$\profile\coloneqq\profiletuple$ with linear preference orders,
  a non-negative integer~$k\in \mathds{N}$,
  a preferred \alternative~$p\in \alternativeset$, 
  and an agenda~$\agenda$ for $\alternativeset$.}
{Is it possible to add a coalition of (that is, a set of) $k$~voters such that $p$ wins under agenda~$\agenda$?}

We find that deciding whether a manipulation is successful is polynomial-time solvable
for both the successive procedure and the amendment procedure. 
However,
our approach to deciding whether the amendment procedure can be 
successfully manipulated is asymptotically
more complex than our approach to deciding the same question 
for the successive procedure.

First, we observe that the manipulators can basically vote in the same way.

\begin{observation}\label{obs:all-manipulators-vote-the-same}
  \obsmanipulatorsvotethesame
\end{observation}

\begin{proof}
  For the successive procedure, 
  if there is a successful manipulation for the preferred alternative~$p$, 
  then requiring all manipulators to rank $p$ in the first position and to rank the other \alternatives in an arbitrary but fixed order also makes $p$ win.

  Now, 
  let $\profile'$ be the manipulated profile, that is, the original profile plus the manipulators.
  Let $X\coloneqq \{x_1, x_2, \ldots, x_s\}$ consist of all alternatives~$x_i \in \alternativeset$ such that there is a round index~$r_i$ where $x_i$ is an $r_i$th-round amendment winner in $\profile'$.
  Assume without loss of generality that for every $2 \le i \le s$,
  alternative~$x_{i}$ beats $x_{i-1}$.
  If $p$ is an amendment winner in $\profile'$,
  then $p=x_s$.
  Furthermore, we can verify that each $x_i$ is an $r_i$th-round amendment
  winner in the original profile plus the $k$ manipulators who all have preference order~$x_s \pref x_{s-1} \pref \ldots \pref x_{1} \pref \setseq{\alternativeset\setminus X}$.
  This implies that $x_s=p$ is an amendment winner.
\end{proof}

We also observe that the best way of manipulating the successive procedure is to let the manipulators rank their most preferred alternative at the first position.

\begin{theorem}\label{thm:manipulation-successive-p}
  \thmmanipulationsuccessive
\end{theorem}

\begin{proof}
  In the proof of Observation~\ref{obs:all-manipulators-vote-the-same},
  we can observe that if a coalition of $k$~\voters can manipulate the successive procedure,
  then by ranking alternative~$p$ in the first position and
  the other \alternatives in an arbitrary but fixed order, $p$ must
  also win. This leads to a linear-time $O((k+n)\cdot m)$ algorithm:
  Let the coalition all vote $p\pref \setseq{\alternativeset\setminus
    \{p\}}$, 
  and check whether $p$
  may win the successive procedure.
  Since a successful manipulation is always possible when $k \ge n+1$,
  the linear running time is indeed $O(n\cdot m)$.
\end{proof}

For the successive procedure, 
how the manipulators should vote does not depend on the sequence of the alternatives in the agenda.
For the amendment procedure, however, a successful manipulation depends greatly on the agenda.
To formalize this idea, we need one more notion:

\begin{definition}
  Let $(\profile\coloneqq \profiletuple, k, p, \agenda)$ be an instance of \probManipulation under the amendment procedure.
  We call an alternative~$c \in \alternativeset$ an
  \emph{\nthmwinner{i}}  
  if adding a coalition of $k$ additional voters to the
  original profile~$\profile$ makes this \alternative the \nthawinner{i}
  under agenda~$\agenda$.
\end{definition}

The crucial point is that an alternative that is \emph{not} an \nthmwinner{i} can never become a \nthmwinner{j} with $j>i$.
Further, it is easy to check whether an \nthmwinner{i} alternative can survive as a manipulated winner of a later round as the following shows.

\begin{lemma}\label{lem:manipulated-winner}
  Let $b$ be the $i$th \alternative in the agenda~$\agenda$, $2\le i\le m$.
  Then,
  \begin{enumerate}
    \item\label{cond:b-ith-mwinner} $b$ is an \nthmwinner{i} if and only if 
    there is an \nthmwinner{(i-1)} $c$ such that
    requiring all manipulators to prefer~$b$ to $c$ makes $b$ beat $c$.  
    \item\label{cond:c-ith-mwinner} an \nthmwinner{(i-1)} $c$ is also an \nthmwinner{i} 
    if and only if requiring all manipulators to prefer~$c$ to $b$ makes $c$ beat $b$.  
\end{enumerate}
\end{lemma}

\begin{proof}
  We only show the first statement as the second one can be shown analogously.
  For the ``only if'' part, assume that $b$ is an \nthmwinner{i} and let $\profile'$ be the manipulated profile.
  Now, we rename and enumerate all \nthawinner{j}s in $\profile'$ with $j \le i-1$ as
  $x_1$, $x_2$, $\ldots$, $x_{s}$ such that for each $1\le \ell \le s-1$,
  alternative~$x_{\ell+1}$ beats $x_{\ell}$ in $\profile'$.
  By definition, $x_{s}$ is the \nthawinner{(i-1)}.
  Now observe that
  $b$ is an \nthawinner{i} in the manipulated profile, where every manipulator has preference order $b \pref x_{s}\pref x_{s-1} \pref \ldots \pref x_1 \pref \setseq{\alternativeset\setminus X}$.
  
  For the ``if'' part, let $c$ be an \nthmwinner{(i-1)} and consider the corresponding manipulators' preference orders~$\pref'_1, \pref'_2, \ldots, \pref'_k$ that only rank the first $(i-1)$ alternatives in the agenda.
  If requiring every manipulator to prefer $b$ to $c$ can make $b$ beat $c$, 
  then at round $i$, 
  $b$ will survive as the \nthawinner{i} when every manipulator~$i$ has a preference order extending~$\pref'_i \cup \{b \pref c\}$.  
  This implies that $b$ is an \nthmwinner{i}.
\end{proof}

By the above lemma, we can indeed compute in quadratic time all \alternatives that can become an amendment winner by adding $k$ additional voters to the profile,
and we can compute the corresponding coalition for each of these \alternatives. 

\begin{theorem}\label{thm:manipulation-amendment-p}
  \thmmanipulationamendment
\end{theorem}

  \begin{algorithm}[t!]
    \SetCommentSty{\color{gray}}
    \SetAlgoVlined
    \SetKwInput{Input}{Input}
    \SetKwBlock{Block}{}{}

    \SetKw{KwNo}{``no''}
    \SetKw{KwYes}{``yes''}
    \SetKw{KwAnd}{and}

    \Input{
      $(\profile = \profiletuple, k, p, \agenda)$ \comm{---  an instance of \probManipulation}
    }
    Set $W_1 \coloneqq \{\text{the first alternative in the agenda~} \agenda\}$\\
    Set $\pref^{b}_1 \coloneqq \emptyset$\\
    \ForEach{round $i \in \{2, 3, \ldots, m\}$}
    {
      \label{alg1:loop-begin}
      Set $W_i \coloneqq \emptyset$\\
      Set $b\coloneqq$ the $i$th \alternative in $\agenda$\\
      \ForEach{\nthmwinner{(i-1)}~$c\in W_{i-1}$}
      { 
        \If{$b \notin W_{i}$ \KwAnd $b$ beats $c$ when all manipulators prefer $b$ to $c$}
        { \label{check:b->c-begin} 
          Add $b$ to the set $W_i$ \\
          Construct the preference order~$\pref^{b}_{i} \coloneqq \{c \pref b\} \cup \pref^{c}_{i-1}$ \label{check:b->c-end}         
        }
        \If{$c$ beats $b$ when all manipulators prefer $c$ to $b$}
        { \label{check:c->b-begin}
          Add $c$ to the set $W_i$  \\
          Construct the preference order~$\pref^{c}_{i} \coloneqq \{c \pref b\} \cup \pref^{c}_{i-1}$ \label{check:c->b-end}
        }
      }
    } \label{alg1:loop-end}
    \caption{Algorithm for computing all manipulated winners for the amendment procedure.}
    \label{alg:manipulation-amendment}
  \end{algorithm}

\begin{proof}
  Based on Lemma~\ref{lem:manipulated-winner} we can build a recursive algorithm 
  that constructs a linear order over the first $i$~\alternatives from
  $\agenda$ for each \nthawinner{i}, $i$ starting from $1$.
  
  We denote by $W_{i}$ the set of all \nthmwinner{i}s.
  For each \nthmwinner{i}~$c$, 
  we denote by $\pref^{c}_{i}$ the preference order over the first $i$ alternatives in the agenda~$\agenda$ 
  such that $c$ becomes an \nthawinner{i} by adding $k$~manipulators with preference order that extend $\pref^{c}_{i}$.

  The approach to computing all manipulated winners is described in Algorithm~\ref{alg:manipulation-amendment}.
  Obviously, for the first round, 
  the set~$W_1$ and its corresponding preference order are computed correctly.
  By Lemma~\ref{lem:manipulated-winner} (\ref{cond:b-ith-mwinner}),
  we know that Steps~(\ref{check:b->c-begin})-(\ref{check:b->c-end}) are correct,
  and by Lemma~\ref{lem:manipulated-winner}(\ref{cond:c-ith-mwinner}),
  we know that Steps~(\ref{check:c->b-begin})-(\ref{check:c->b-end}) are correct.
  Since each alternative from $W_{m}$ is a last-round manipulated winner,
  we answer ``yes'' to the instance $(\profile = \profiletuple, k, p, \agenda)$ if $p \in W_{m}$ and ``no'' otherwise.

  Now, we come to the running time.
  First of all, 
  for each two distinct alternatives~$b, c$, 
  we check whether adding $k$ manipulators can make $b$ beat~$c$; 
  let the Boolean variable~$T(b,c)$ have value one if this is the case and zero otherwise.
  Computing all these Boolean values runs in $O((k+n)\cdot m^2)$.
  
  Then, to compute $W_i$, $2\le i \le m$, 
  we check $T(b,c)$ for every alternative~$c$ in $W_{i-1}$ and for the $i$th alternative~$b$ in the agenda~$\agenda$.
  Thus computing all $W_i$ can be done in $O(m^2)$ time.
  Since a successful manipulation is always possible when $k \ge n+1$,
  the total running time is indeed $O(n\cdot m^2)$.
\end{proof}

In the \probWManipulation problem, the \voters of the coalition 
also come with integer weights. 
However, we remark here that the weighted and non-weighted cases are
equivalent because of Observation~\ref{obs:all-manipulators-vote-the-same}.
Observe that if the sum of the weights is greater than the number of the voters in the original profile, 
then a successful manipulation is always possible. 
Thus, we can conclude the following.

\begin{corollary}\label{cor:wmanipulation-p}
  \corwmanipulation
\end{corollary}

\section{Possible/Necessary Winner}
\label{sec:possible-necessary}

We typically have partial knowledge about how the \voters will vote,
and about how the agenda will be ordered.
Nevertheless, we might be interested
in what may or may not
be the final outcome. Does our favorite \alternative stand any
chance of winning? Is it inevitable
that the government's \alternative
will win? Is there an agenda under which
our \alternative can win? 
Hence, we consider 
the question of which \alternative possibly or necessarily
wins. 

\probDef{\textsc{Possible } (resp.\ \textsc{Necessary) Winner}}
{ A preference profile~$\profile\coloneqq\profiletuple$, 
  a preferred \alternative~$p\in \alternativeset$, and a partial agenda~$\partialagenda$.}
{Can $p$ win in a (resp. every) completion of the
  profile~$\profile$ for an (resp.\ every) agenda which completes~$\partialagenda$?}

An upper bound for the computational complexity of both problems is easy to see:
\probPossibleWinner (resp.\ \probNecessaryWinner) for both the successive procedure and the amendment procedure is contained in $\NP$ (resp.\ in $\coNP$) because
determining a winner for both procedures is polynomial-time solvable.
Thus, in order to show the $\NP$-completeness (resp.\ $\coNP$-completeness),
each time we only need to show the  $\NP$-hardness (resp.\ $\coNP$-hardness).

\subsection{Possible Winner}
Our first two results imply that as soon as the voters may have non-linear preference orders,
deciding who may be a possible winner is $\NP$-hard even for a fixed agenda. 
The last result shows that
if the number~$m$ of alternatives is a constant, 
then this problem becomes polynomial-time solvable and the degree of the polynomial does not depend on $m$.

\begin{theorem}
  \label{thm:possiblewinner-successive-nph}
  \thmpossiblewinnersuccessive
\end{theorem}

\newcommand{\restseta}{\ensuremath{\alternativeset \setminus (\{p, d\} \cup I(u_i))}}
\newcommand{\restsetb}{\ensuremath{\alternativeset \setminus \{p, d\}}}

\begin{proof}
  We show the $\NP$-hardness by reducing from the $\NP$-complete \probIS problem in polynomial time.
  Given an undirected graph~$G=(U,E)$ and an integer~$h$, \probIS asks whether there is
  an \emph{independent set} of size at least~$h$, that is,
  a subset of at least $h$~vertices such that no two of them are adjacent to each other.
  We will give a concrete example for the reduction right after the proof.

  Let $(G=(U, E),h)$ be an instance of the \probIS problem, where
  $U=\{u_1,\dots,u_r\}$ denotes the set of vertices and
  $E=\{e_1,\dots,e_s\}$ denotes the set of edges.
  We assume that $r\ge 3$ and $2 \le h \le r-1$.
  We construct a \probPossibleWinner instance~$((\alternativeset, \voterset),p,\partialagenda)$ as follows.
  The set~$\alternativeset$ of alternatives consists of
  the preferred alternative~$p$,
  one \emph{dummy alternative}~$d$, 
  and for each edge~$e_j \in E$ one \emph{edge alternative}~$c_j$:
  \[ \alternativeset \coloneqq \{p, d\} \cup \{c_j \mid e_j \in E\} \text{.}\]
  We construct three groups of voters, 
  where only the first group of voters has partial orders while the remaining two groups have linear orders.
  \begin{enumerate}
  \item For each vertex~$u_i \in U$, we construct a \emph{vertex voter}~$v_i$ with a partial preference order specified by
  \[\setseq{I(u_i)} \pref p \pref \setseq{\restseta} \text{ and } d \pref \setseq{\restseta}\text{,} \]
  where $I(u_i)$~denotes the set of edge alternatives corresponding to edges incident to vertex~$u_i$.
  Briefly put,  vertex voter~$v_i$ prefers all ``incident'' edge alternatives to all other alternatives (including $p$) except $d$ 
  but he thinks that every ``incident'' edge alternative and $p$ are incomparable to $d$.
  \item We construct $h-2$ \emph{auxiliary voters} with the same complete preference order
  \[ \setseq{\restsetb} \pref d \pref p\text{.} \]
  \item We construct another $r-h-1$ auxiliary voters with the same complete preference order
  \[ \setseq{\restsetb} \pref p \pref d\text{.} \]
  \end{enumerate}
  Note that we have constructed a total of $2r-3$ voters. 
  Thus, to be a majority winner,
  an alternative needs to be preferred at the first place by at least $r-1$ voters.
  
  Let the agenda~$\partialagenda$ be the linear order~$c_1 \agendapref c_2 \agendapref \ldots \agendapref c_{s} \agendapref p \agendapref d$.
  This completes the construction which can clearly be computed in polynomial time.

  Before we give a formal correctness proof, let us briefly sketch the idea.
  Our construction ensures that in order to beat~$d$ in the final round
  of the procedure we have to put~$p$ (and by the construction of the vertex voters, also $I(u_i)$)
  in front of $d$ in at least $h$~vertex voters' preference orders.
  However, to prevent any edge alternative~$c_j$ from being an amendment winner in some 
  earlier round of the procedure, we cannot put any edge alternative in front of~$d$ more than once.
  We will see that this is only possible if the vertices corresponding to the
  voters for which we put $p$ in front of~$d$ form an independent set of size $h$.

  Now, we show that $G$~has an independent set of size at least~$h$ if and only if
  $p$ can possibly win in the constructed instance.

  For the ``only if'' part, assume that $G$ admits an independent set~$U'\subseteq U$ of size at least~$h$.
  Then, we complete the partial preference orders of the vertex voters to make~$p$ a successive winner as follows.
  \begin{enumerate}
    \item For each vertex~$u_i \in U'$ belonging to the independent set, let
    voter~$v_i$ have the preference order
    \[\setseq{I(u_i)} \pref p \pref d \pref \setseq{\restseta}\text{.}\]
    \item  For each vertex~$u_i \notin U'$ \emph{not} belonging to the independent set, let
    voter~$v_i$ have the preference order
    \[d \pref \setseq{I(u_i)} \pref p \pref \setseq{\restseta}\text{.}\]
  \end{enumerate}
  Since $U'$ is an independent set, 
  every edge \alternative is preferred to $d$ by at most one vertex voter.
  Together with the remaining $r-3$ auxiliary voters,
  every edge \alternative is preferred to $d$ by at most $r-2$ voters, 
  causing each edge \alternative to be deleted if it is considered prior to $d$ in an agenda (note that we have $2r-3$ voters). 
  Hence, by our agenda where all edge \alternatives are in front of $d$, 
  each edge \alternative is indeed deleted.
  Now, we come to the last but one round.
  Since the independent set~$U'$ has size at least $h$,
  at least $h$ vertex voters prefer $p$ to $d$.
  Then, $p$ will beat $d$, because the first $r-h-1$ auxiliary voters prefer $p$ to $d$,
  making $p$ a successive winner.

  For the ``if'' part, assume that $p$ can possibly become a successive winner which means that 
  one can complete the vertex voters' preference orders, ensuring $p$'s victory.
  Let $V'$ be the set of vertex voters that prefer $p$ to $d$ in such a completion.
  Since a total of $r-h-1$ auxiliary voters prefer $p$ to $d$,
  in order to make $p$ beat $d$ at the last but one round, 
  $|V'|$ must be at least $h$ (note that the majority quota is $r-1$).
  Now, consider the vertex subset~$U'$ that corresponds to $V'$. 
  We show that $U'$ is an independent set, that is, 
  no two vertices in $U'$ are adjacent.
  Suppose for the sake of contradiction that 
  there are two adjacent vertices~$u, u' \in U'$; denote the edge by $e_j$.
  Then, back to our completed profile, 
  since $p$ is a successive winner, 
  one must come to the round where the corresponding edge \alternative~$c_j$ is considered.
  Since all edge \alternatives with a lower index are already deleted, 
  by construction of the preference orders,
  a total of $(r-1)$ voters (the two vertex voters corresponding to $u, u'$ and all $r-k-3$ auxiliary voters) rank $c_j$ at the first position in the profile restricted to the alternatives~$c_j, c_{j+1}, \ldots, c_{s}, p, d$.
  This, however, will make $c_j$ win---a contradiction.
\end{proof}

\tikzstyle{sol}=[circle, minimum size=15pt, inner sep = 2pt, draw=black!50, fill=black!10]
\tikzstyle{vertex}=[circle, minimum size=15pt, inner sep = 2pt, draw=black!50, fill=white]
\tikzstyle{line}=[draw=black!50,-]
\begin{figure}[t!]
  \centering
  \newcommand{\sss}{\ensuremath \scriptscriptstyle}
\begin{tikzpicture}[auto, >=stealth']
  %
  %
  %
  \node[vertex] at (0,0) (u1) {$u_{\sss 1}$};
  \node[sol, below = 5ex of u1, xshift = -19.5ex] (u3) {$u_{\sss 3}$};
  \node[sol, below = 5ex of u1, xshift = -6.6ex] (u4) {$u_{\sss 4}$};
  \node[sol, below = 5ex of u1, xshift = 6.5ex] (u5) {$u_{\sss 5}$};
  \node[sol, below = 5ex of u1, xshift = 19.5ex] (u6) {$u_{\sss 6}$};
  \node[vertex, below = 5ex of u4, xshift = 6.5ex] (u2) {$u_{\sss 2}$};

  \draw[line] (u1) edge node[above,pos=.6] {$e_{\sss 1}$} (u3);
  \draw[line] (u1) edge node[above,pos=.7, xshift=-2pt] {$e_{\sss 3}$} (u4);
  \draw[line] (u1) edge node[above,pos=.7, xshift=2pt] {$e_{\sss 5}$} (u5);
  \draw[line] (u1) edge node[above,pos=.6] {$e_{\sss 7}$} (u6);

  \draw[line] (u2) edge node[below,pos=.6] {$e_{\sss 2}$} (u3);
  \draw[line] (u2) edge node[below,pos=.7, xshift=-2pt] {$e_{\sss 4}$} (u4);
  \draw[line] (u2) edge node[below,pos=.7, xshift=2pt] {$e_{\sss 6}$} (u5);
  \draw[line] (u2) edge node[below,pos=.6] {$e_{\sss 8}$} (u6);

\end{tikzpicture}
\caption{An undirected graph with $6$ vertices and $8$ edges. The graph has an independent set of size~$4$ (filled in gray) and has a vertex cover of size~$2$ (filled in white).}
\label{fig:graph}
\end{figure}
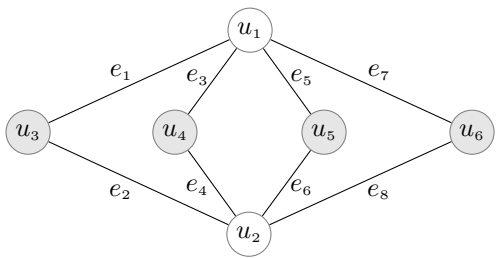

\begin{table}[t!]
  \centering
\begin{tabular}{@{}ll@{}}
  $\text{Voter } v_1\colon$ & $\{\,c_1 \spref c_3 \spref c_5 \spref c_7 \spref p,\; d \, \}  \spref c_2 \spref c_4 \spref c_6 \spref c_8$ \\
  $\text{Voter } v_2\colon$ & $\{\,c_2 \spref c_4 \spref c_6 \spref c_8 \spref p,\;  d \, \}  \spref c_1 \spref c_2 \spref c_3 \spref c_4$ \\
  $\text{Voter } v_3\colon$ & $\{\,c_1 \spref c_2 \spref p,\;  d \, \} \spref c_3 \spref c_4 \spref c_5 \spref c_6 \spref c_7 \spref c_8$ \\
  $\text{Voter } v_4\colon$ & $\{\,c_3 \spref c_4 \spref p,\;  d \, \} \spref c_1 \spref c_2 \spref c_5 \spref c_6 \spref c_7 \spref c_8$ \\
  $\text{Voter } v_5\colon$ & $\{\,c_5 \spref c_6 \spref p,\;  d \, \} \spref c_1 \spref c_2 \spref c_3 \spref c_4 \spref c_7 \spref c_8$ \\
  $\text{Voter } v_6\colon$ & $\{\,c_7 \spref c_8 \spref p,\;  d \, \} \spref c_1 \spref c_2 \spref c_3 \spref c_4 \spref c_5 \spref c_6$ \\
  $\text{One voter} \colon$ & $c_1 \spref c_2 \spref c_3 \spref c_4 \spref c_5 \spref c_6 \spref c_7 \spref c_8 \spref p \spref d$ \\ 
  $\text{Two voters}\colon$ & $c_1 \spref c_2 \spref c_3 \spref c_4 \spref c_5 \spref c_6 \spref c_7 \spref c_8 \spref d \spref p$\\\bottomrule
  $\text{Agenda } \partialagenda\colon$ &$c_1 \agendapref c_2 \agendapref c_3 \agendapref c_{4} \agendapref c_5 \agendapref c_6 \agendapref c_7 \agendapref c_8 \agendapref p \agendapref d$ 
\end{tabular}
\caption{The instance~($\profiletuple, p, \partialagenda$) for the problem \probPossibleWinner and for the successive procedure, obtained from the graph in Figure~\ref{fig:graph} and $h=4$,
  where $\alternativeset = \{c_1, c_2, c_3, c_4, c_5, c_6, c_7, c_8, p, d\}$,
  the voter set~$\voterset$ and the corresponding preference orders and the agenda are depicted above.
  By the construction of the agenda, one may notice that in order to let $p$ beat $d$ in the last two rounds,
  at least $h=4$ vertex voters must rank $p \pref d$.
  But, in order not to let an edge \alternative obtain too much ``support'',
  at most one of its ``incident'' vertex voters can prefer $p$ to $d$.
}
\label{tab:reduced-instance-possible-successive}
\end{table}

We illustrate the above $\NP$-hardness reduction through an example.
Figure~\ref{fig:graph} depicts an undirected graph $G$ with $6$ vertices and $8$ edges.
We set $h\coloneqq 4$.
The gray vertices form an independent set of size $4$.
Then a reduced instance for \probPossibleWinner and for the successive procedure will have $6+(4-2)+(6-4-1) = 9$ \voters and $2+8=10$ \alternatives.
This instance can be found in Table~\ref{tab:reduced-instance-possible-successive}.

Now, we show that for the amendment procedure, 
\probPossibleWinner remains intractable.

\begin{theorem}\label{thm:possiblewinner-amendment-nph}
  \thmpossiblewinneramendment
\end{theorem}

\renewcommand{\restseta}{\ensuremath{\alternativeset \setminus (\{p,b,d\} \cup I(u_i))}}
\renewcommand{\restsetb}{\ensuremath{\alternativeset \setminus \{p,b,d\}}}
\newcommand{\wholeset}{\ensuremath{\{p, b, d\}\cup \{c_j \mid e_j \in E\}}}

\begin{proof}
  We show the $\NP$-hardness by reducing from the $\NP$-complete \probVC problem in polynomial time.
  Given an undirected graph~$G=(U,E)$ and an integer~$h$, \probVC asks whether there is
  a \emph{vertex cover} of size at most~$h$, that is,
  a subset of at most $h$~vertices whose removal destroys all edges.
  We will give a concrete example for the reduction behind the proof.

  Let $(G=(U, E),h)$ be a \probVC instance, where 
  $U=\{u_1,\dots, u_r\}$ denotes the set of vertices and
  $E=\{e_1,\dots,e_s\}$ denotes the set of edges.
  We assume that $r \ge 1$ and $1 \le h \le r-1$.

  We construct a \probPossibleWinner instance~$(\profiletuple,p,\partialagenda)$ as follows.
  The set~$\alternativeset$ of alternatives consists of
  the preferred alternative~$p$,
  one \emph{helper alternative}~$b$, 
  one \emph{dummy alternative}~$d$, 
  and for each edge~$e_j \in E$ one \emph{edge alternative}~$c_j$:
  \[ \alternativeset \coloneqq \{p,b,d\} \cup \{c_j \mid e_j \in E\}\text{.} \]
  
  We construct three groups of voters for the voter set~$V$,
  where only the first group of voters has partial orders while the remaining two groups of voters have linear orders. 
  We let $\setseq{\restsetb}$ denote the order~$c_1 \pref c_2 \pref \ldots \pref c_s$.
  \begin{enumerate}
    \item For each vertex~$u_i$, 
    we construct a \emph{vertex voter}~$v_i$ with a partial order specified by 
    \begin{align*}
      &\setseq{\restseta} \pref I(u_i) \pref d \text{ and }\\
      &\setseq{\restseta} \pref I(u_i) \pref b \pref p \text{,}
    \end{align*}
    where $\setseq{\restseta}$ denotes the order derived from $\setseq{\restsetb}$ by removing $u_i$'s ``incident'' edge \alternatives.
    Briefly put,
    vertex voter~$v_i$ prefers all ``non-incident'' edge alternatives to the remaining ones,
    but he regards every ``incident'' edge alternative and $d$ to be incomparable to 
    $b$ and to $p$.
    \item We construct $r-h-1$ \emph{auxiliary voters} with the same linear order
    \[ p \pref \setseq{\restsetb} \pref b \pref d\text{.} \]
    \item Finally, we construct another $h$ \emph{auxiliary voters} with the same linear order
    \[ p \pref d \pref \setseq{\restsetb} \pref b\text{.} \]
  \end{enumerate}

  Note that we have constructed a total of $2r-1$ voters.
  Thus, given two alternatives~$a, a'$,
  we have that $a$ beats $a'$ if and only if at least $r$ voters prefer $a$ to $a'$.
  
  Let the partial agenda~$\partialagenda$ be the linear order~$ b \agendapref d \agendapref p \agendapref c_s \agendapref c_{s-1} \agendapref \ldots \agendapref c_1\text{.}$
  This completes the construction which can clearly be computed in polynomial time.
 
  \newcommand{\sss}{\ensuremath }
  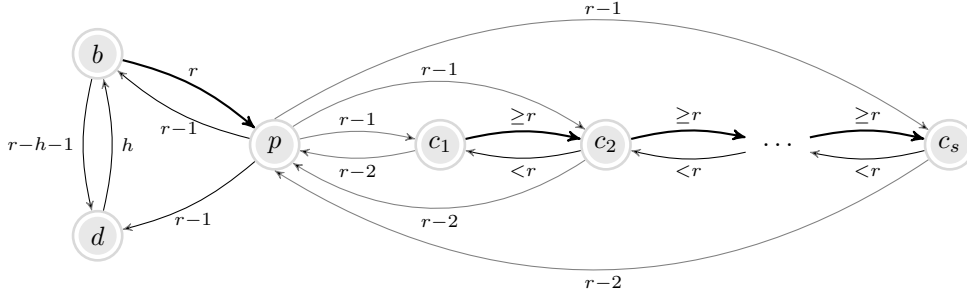
\begin{figure}
    \centering
    \tikzstyle{majarr}=[draw=black,thick,->]
    \tikzstyle{minarr}=[draw=black!50,->]
    \tikzstyle{minarrb}=[draw=black!50,<-]
    \tikzstyle{alter}=[circle, minimum size=18pt, inner sep = 2.5pt, draw=black!15, ultra thick, fill=black!9]
    \tikzstyle{innercircle}=[circle, minimum size=16pt, ultra thick, draw=white]
  \begin{tikzpicture}[auto, >=stealth']
    \node[alter] at (-.55, 1) (b) {$b$};
    \node[alter, right = 11ex of b, yshift=-8ex] (p) {$p$};
    \node[alter, left = 11ex of p, yshift=-8ex] (d) {$d$};
    \node[innercircle] at (b) {};
    \node[innercircle] at (d) {};
    \node[innercircle] at (p) {};
    
    \draw[majarr] (b) edge[bend left=13] node[midway, anchor=south] {$\scriptstyle r$} (p);
    \draw[minarr] (p) edge[bend left=13] node[midway, anchor=north] {$\scriptstyle r-1$} (b);
    \draw[minarr] (p) edge[bend left=13] node[midway, anchor=north] {$\scriptstyle r-1$} (d);

    \draw[minarr] (b) edge[bend right=13] node[midway, anchor=east] {$\scriptstyle r-h-1$} (d);
    \draw[minarr] (d) edge[bend right=13] node[midway, anchor=west] {$\scriptstyle h$} (b);
    
    \node[alter, right = 10ex of p] (e1)  {$c_{\sss 1}$};
    \node[alter, right = 10ex of e1] (e2) {$c_{\sss 2}$};
    \node[right = 10ex of e2] (lds) {$\;\ldots\;$};
    \node[alter, right = 10ex of lds] (em) {$c_{\sss s}$};

    \node[innercircle] at (e1) {};
    \node[innercircle] at (e2) {};
    \node[innercircle] at (em) {};
    
    \draw[majarr] (e1) edge[bend left=13] node[midway, anchor=south, yshift=-1pt] {$\scriptstyle \ge r$} (e2);
    \draw[majarr] (e2) edge[bend left=13] node[midway, anchor=south, yshift=-1pt] {$\scriptstyle \ge r$} (lds);
    \draw[majarr] (lds) edge[bend left=13] node[midway, anchor=south, yshift=-1pt] {$\scriptstyle \ge r$} (em);

    \draw[minarrb] (e1) edge[bend right=13] node[midway, anchor=north, yshift=1pt] {$\scriptstyle < r$} (e2);
    \draw[minarrb] (e2) edge[bend right=13] node[midway, anchor=north, yshift=1pt] {$\scriptstyle < r$} (lds);
    \draw[minarrb] (lds) edge[bend right=13] node[midway, anchor=north, yshift=1pt] {$\scriptstyle < r$} (em);

    \draw[] (p) edge[bend left=13,minarr] node[midway, anchor=south, yshift=-2pt] {$\scriptstyle r-1$} (e1);
    \draw[] (e1) edge[bend left=13,minarr] node[midway, anchor=north, yshift=1pt] {$\scriptstyle r-2$} (p);

    \draw[] (p.north east) edge[bend left=33,minarr] node[midway, anchor=south, yshift=-2pt] {$\scriptstyle r-1$} (e2);
    \draw[] (e2) edge[bend left=33,minarr] node[midway, anchor=north, yshift=1pt] {$\scriptstyle r-2$} (p.south east);

    \draw[] (p.north) edge[bend left=33,minarr] node[midway, anchor=south, yshift=-2pt] {$\scriptstyle r-1$} (em);
    \draw[] (em) edge[bend left=33,minarr] node[midway, anchor=north, yshift=1pt] {$\scriptstyle r-2$} (p.south);
 \end{tikzpicture}
 \caption{A weighted majority graph (not including all arcs) for the profile ($2r-1$ \voters and $s+3$ \alternatives) in the reduced instance for \probPossibleWinner with the amendment procedure (see Definition~\ref{def:majority graph} for the corresponding definition).
   Note that we use the same symbol for both the alternative and its corresponding vertex because it will always be clear from the context which one we mean.
   For instance, there is an arc from $b$ to $p$ with weight $r$ because exactly $r$ voters prefer $b$ to $p$.
   They all come from the first group of voters.
   We mark an arc~$(a,a')$ bold if there is majority of voters preferring $a$ to $a'$.
   For the sake of brevity, some (irrelevant) arcs are omitted.
   By `$\ldots$'' we refer to the remaining edge vertices in increasing order.
  }
  \label{fig:maj_graph-possible-amendent-instance}
  \end{figure}

  First, consider Figure~\ref{fig:maj_graph-possible-amendent-instance} which illustrates the corresponding majority graph for the profile.
  Let us briefly sketch the idea with the help of the majority graph:
  By the constructed agenda, 
  in order to become an amendment winner,
  our preferred alternative~$p$ has to beat every edge alternative~$c_j$, $1\le j \le s$.
  By the auxiliary voters' preference orders (also see the corresponding weighted majority graph in Figure~\ref{fig:maj_graph-possible-amendent-instance}),
  exactly $r-1$ voters prefer $p$ to $c_j$, $1\le j \le s$.
  Thus, we have to put both $p$ (and by the construction of the vertex voters, also~$b$) in front of $c_j \pref d$
  in the preference order of at least one vertex voter that corresponds to a vertex incident to the edge~$e_j$.
  This implies that the vertices corresponding to the
  voters for which we put both $p$ and $b$ in front of~$d$ form a vertex cover.
  Furthermore, since $b$ beats $p$ (see the corresponding arc weight in the majority graph),
  $d$ has to beat $b$ in the first round of the procedure.
  Since all $r-h-1$ voters from the second group prefer~$b$ to $d$,
  we are only allowed to put $b \pref p$ in front of $d$ in at most $h$ vertex voters' preference orders.
  Thus, the vertex cover is of size at most~$h$. 

  Now, we show that $G$ has a vertex cover of size at most $h$ if and only if
  $p$ can possibly win the amendment procedure. 

  For the ``only if'' part, 
  assume that there is a vertex cover~$U'$ of size at most~$h$.
  Then, we complete the vertex voters' preference orders to make~$p$ an amendment winner as follows.
  \begin{enumerate}
    \item For each vertex $u_i\in U'$ belonging to the vertex cover, 
    we let $v_i$ have the preference order
    \[\setseq{\restseta} \pref b \pref p \pref \setseq{I(u_i)} \pref d \text{.}\]
    \item For each vertex $u_i\notin U'$ \emph{not} belonging to the vertex cover, 
    we let $v_i$ have the preference order
    \[\setseq{\restseta} \pref \setseq{I(u_i)} \pref d \pref b \pref p \text{.}\]
  \end{enumerate}
 
  Since $|U \setminus U'| \ge r - h$,
  at least $r-h$ vertex voters prefer $d$ to $b$.
  Together with the $h$ voters from the third group, which prefers $d$ to $b$,
  $d$ beats $b$ and survives as the second round winner.
  Since we assume that $h \ge 1$, 
  at least one additional vertex voter prefers $p$ to $d$.
  Together with all $r-1$ voters from the second and the third group that prefer $p$ to $d$,
  $p$ beats $d$ and survives as the third round winner.
  Since $U'$ is a vertex cover,
  for each edge \alternative~$c_j$, $1\le j\le s$,
  there is at least one vertex voter preferring $p$ to $c_j$.
  This implies that $p$ beats $c_j$ (note that all $r-1$ auxiliary voters prefer $p$ to $c_j$).
  Thus, $p$ is an amendment winner.

  For the ``if'' part, assume that $p$ can possibly become an amendment winner which means that 
  one can complete the vertex voters' preference orders, ensuring $p$'s victory.
  Let $V'$ be the set of vertex voters that prefer $p$ to $c_j$ where $c_j$ is an arbitrary edge \alternative.
  As we already noticed,
  since $p$ is in front of every edge \alternative in the agenda,
  $p$ must beat every edge \alternative.
  Thus, for each edge \alternative~$c_j$,
  there is at least one vertex voter~$v_i \in V'$ who prefers $p$ to $c_j$.
  By the construction of our preference orders, 
  this means that $V'$ corresponds to a vertex cover.
  
  Now, observe that every vertex voter in $V'$ also prefers $b$ to $d$ (because he originally prefers $b$ to $p$).
  Since all $r-h-1$ auxiliary voters from the second group prefer $b$ to $d$ and 
  since $p$ can only possibly win if $b$ does not survive the second round (because $b$ beats $p$), 
  there are at most $h$ voters in $V'$.
  This further implies that the vertex set corresponding to $V'$ is a vertex cover of size at most $h$.
\end{proof}

\begin{table}
  \centering
\begin{tabular}{@{}ll@{}}
  $\text{Voter } v_1\colon$ & $c_2 \spref c_4 \spref c_6 \spref c_8 \spref \{\,c_1 \spref c_3 \spref c_5 \spref c_7 \spref d,\; b \spref p\, \}$ \\
  $\text{Voter } v_2\colon$ & $c_1 \spref c_2 \spref c_3 \spref c_4 \spref \{\,c_2 \spref c_4 \spref c_6 \spref c_8 \spref d,\; b \spref p\, \}$ \\
  $\text{Voter } v_3\colon$ & $c_3 \spref c_4 \spref c_5 \spref c_6 \spref c_7 \spref c_8 \spref \{\,c_1 \spref c_2 \spref d,\; b \spref p\, \} $ \\
  $\text{Voter } v_4\colon$ & $c_1 \spref c_2 \spref c_5 \spref c_6 \spref c_7 \spref c_8 \spref \{\,c_3 \spref c_4 \spref d,\; b \spref p\, \} $ \\
  $\text{Voter } v_5\colon$ & $c_1 \spref c_2 \spref c_3 \spref c_4 \spref c_7 \spref c_8 \spref \{\,c_5 \spref c_6 \spref d,\; b \spref p\, \} $ \\
  $\text{Voter } v_6\colon$ & $c_1 \spref c_2 \spref c_3 \spref c_4 \spref c_5 \spref c_6 \spref \{\,c_7 \spref c_8 \spref d,\; b \spref p\, \} $ \\
  $\text{Three voter} \colon$ & $p \spref c_1 \spref c_2 \spref c_3 \spref c_4 \spref c_5 \spref c_6 \spref c_7 \spref c_8 \spref b \spref d$ \\ 
  $\text{Two voters}\colon$ & $p \spref d \spref c_1 \spref c_2 \spref c_3 \spref c_4 \spref c_5 \spref c_6 \spref c_7 \spref c_8 \spref b$\\\bottomrule
  $\text{Agenda } \partialagenda\colon$ &$b \agendapref d \agendapref p \agendapref c_1 \agendapref c_2 \agendapref c_3 \agendapref c_{4} \agendapref c_5 \agendapref c_6 \agendapref c_7 \agendapref c_8$ 
\end{tabular}
\caption{The instance~($\profiletuple, p, \partialagenda$) for \probPossibleWinner with the amendment procedure where $\alternativeset = \{c_1, c_2, c_3, c_4, c_5, c_6, c_7, c_8, p, d\}$, 
  obtained from the graph in Figure~\ref{fig:graph} and $h=2$,
  the voter set~$\voterset$ and the corresponding preference orders and the agenda are depicted above.
  By the construction of the agenda, one may notice that in order to let $p$ become an amendment winner,
  at most $h=2$ vertex voters can rank $b\pref p \pref c_j \pref d$ for some $c_j$.
  In order no beat every edge \alternative~$c_j$,
  at least one ``incident'' voter must rank $b \pref p \pref c_j$.
}
\label{tab:reduced-instance-possible-amendment}
\end{table}

We illustrate the above $\NP$-hardness reduction through an example.
Let us consider the undirected graph depicted in Figure~\ref{fig:graph} again.
We set $h\coloneqq 2$ and the corresponding vertex cover consists of the white vertices.
Then a reduced instance for \probPossibleWinner and for the successive procedure will have $6+(6-2-1) + 2 = 11$ \voters and $3+8=11$ \alternatives.
This instance can be found in Table~\ref{tab:reduced-instance-possible-amendment}.

\ \\[1ex]
We have just shown that it is $\NP$-hard to decide whether \alternative~$p$ can \emph{possibly} be a successive (or an amendment) winner under a fixed agenda~(see Theorems~\ref{thm:possiblewinner-successive-nph} and \ref{thm:possiblewinner-amendment-nph}).
In both $\NP$-hardness reductions, however, the number~$m$ of alternatives and the number~$n$ of voters are unbounded.
We will show in the following that
\probPossibleWinner is \emph{fixed-parameter tractable} ($\FPT$, 
see the textbooks of \citet{CyFoKoLoMaPiPiSa2015,DF13,FG06,Nie06} for more information on parameterized complexity theory) with respect to the parameter ``number~$m$ of \alternatives''.
That is, if $m$ is a constant,
then it is solvable in $f(m)\cdot \mathrm{poly}(|I|)$ time,
where $|I|$ is the size of a given instance of the problem and $f$ is a computable function that solely depends  on $m$.

First, we 
reduce \probPossibleWinner to a special case of the integer linear programming (ILP) problem~\citep{Len83}.
In this special case, each input instance has $f_1(m)$ variables and $f_2(m)$~constraints,
and the maximum of the absolute values of the coefficients and of the constant terms is $O(m\cdot n)$~(recall that $m$ denotes the number of alternatives and $n$ the number of voters),
where $f_1$ and $f_2$ are two functions that only depend on $m$.

\begin{theorem}\label{thm:possiblewinner-ilp}
  \thmpossiblewinnerilp
\end{theorem}

\begin{proof}
  The general idea is to ``guess'' the completion of the agenda and the sequential outcome of the voting regarding this agenda such that our preferred~$p$ may win and use integer linear programs (ILP) to check whether the guess is valid.

  First, we introduce some notation for the description of our ILP for both procedures.
  For each partial order~$\pref$, we write $N(\pref)$ to denote the number of voters with partial order~$\pref$ in the original profile and we write $C(\pref)$ to denote the set of all possible linear orders completing $\pref$.
  For instance, if we have three alternatives~$a, b, c$,
  then $C(\{a, b\}\pref b) = \{a\pref b \pref c, b \pref a\pref c\}$ and $C(a \pref b \pref c)=\{a\pref b \pref c\}$.
  Accordingly, for each linear order~$\pref^{*}$, we write $C^{-1}(\pref^*)$ to denote the set of all partial orders that $\pref^*$ can complete.

  For each partial order~$\pref$ and each possible completion~$\pref^* \in C(\pref)$ of $\pref$,
  we introduce an integer variable~$x(\pref, \pref^*)$.
  It denotes the number of voters with partial order~$\pref$ in the input profile that will have linear
  order~$\pref^*$ in this profile's completion.
  We use these variables for both the successive and the amendment procedures.

  \newcommand{\tsumq}[3]{\quad\sum_{{\mathclap{\substack{#1\\#2}}}}{#3\quad}}
  \newcommand{\tsumlq}[3]{\quad\sum_{\mathclap{\substack{#1\\#2}}}{#3}}
  \newcommand{\tsumrq}[3]{\sum_{\mathclap{\substack{#1\\#2}}}{#3}\quad}

  \paragraph{Successive procedure.}
  Suppose that $\agenda\coloneqq a_1\agendapref a_2 \agendapref \ldots \agendapref a_y \agendapref \ldots \agendapref a_{m}$ is the guessed agenda with $a_y$ being our preferred alternative~$p$.
  We need one more notation: For each alternative~$a_i$, $1\le i \le m-1$, 
  let $F(a_i)$ denote the set of linear orders that prefer $a_i$ to the rest~$\{a_{i+1}, a_{i+2}, \ldots, a_{m}\}$.
  Now, the crucial point is that our preferred alternative~$p$ is a successive winner under agenda~$\agenda$ if and only if  the following two conditions hold.
  \begin{enumerate}
    \item\label{cond2:all-prev-alts-lose} The profile can be completed such that \emph{no} alternative~$a_i$, $1\le i \le y-1$, is an $\{a_i, a_{i+1}, \ldots, a_{m}\}$-majority winner.
    \item \label{cond3:p-wins} Alternative~$p$ is an $\{a_y, a_{y+1}, \ldots, a_{m}\}$-majority winner.
  \end{enumerate}
  We can describe these two conditions via a formulation for the ILP feasibility problem.
  \allowdisplaybreaks
  \begin{align}
   \label{cons1:exact-completion} \sum_{\mathclap{\substack{\pref^* \in C(\pref)}}}{x(\pref, \pref^*)} & = N(\pref), && \forall \text{ partial orders } \pref, \\
   \label{cons2:all-prev-alts-lose} \tsumlq{\pref^* \in F(a_i)}{\pref \in C^{-1}(\pref^*)}{x(\pref, \pref^*)} & \le \frac{n}{2}, && \forall 1 \le i \le y-1, \\
   \label{cons3:p-wins} \tsumlq{\pref^* \in F(p)}{\pref \in C^{-1}(\pref^*)}{x(\pref, \pref^*)} &> \frac{n}{2}\text{.} & &
  \end{align}

  First, Constraint~(\ref{cons1:exact-completion}) ensures that each voter's partial order is completed to exactly one linear order.
  Second, Constraints~(\ref{cons2:all-prev-alts-lose}) and (\ref{cons3:p-wins}) correspond to Conditions~\ref{cond2:all-prev-alts-lose} and \ref{cond3:p-wins}, respectively.
  The correctness of our ILP thus follows.
  
  For the running time, first of all, we brute-force search into all $m!$ possible completions of the input agenda. 
  Then,  for each of these completions, we run an ILP with at most $m! \cdot 2^{m^2}$ variables and at most $2^{m^2}+m$ constraints and where the absolute value of each coefficient and the constant term is at most~$n$.

  \paragraph{Amendment procedure.}
  As for the successive procedure, we guess a completion of the input agenda.
  Apart from this, we also guess the amendment winners in each round.
  Suppose that $\agenda\coloneqq a_1\agendapref a_2 \agendapref \ldots \agendapref a_y \agendapref \ldots \agendapref a_{m}$ is the guessed agenda with $a_y$ being our preferred alternative~$p$,
  and that for each $2\le i \le y-1$, alternative~$b_i\in \{a_1, a_2, \ldots, a_i\}$ is the guessed \nthawinner{i}. (the first round amendment winner is by definition $a_1$).
  Note that guessing the amendment winner of a round higher than or equal to $y$ is not necessary since the goal is to make $p$ the \nthawinner{z} with $z\ge y$.
  Before we turn to the ILP formulation, we first check in $O(m)$ time whether the guessed amendment winners are valid, that is,
  \begin{quote}
    for each $i \in \{2,3,\ldots,  y-1\}$ with $b_{i} \neq b_{i-1}$, it must hold that $b_i = a_i$.
  \end{quote}
  The reason of this sanity check is that under the amendment procedure, 
  an alternative~$c$ can be an amendment winner of consecutive rounds and 
  the only cause of a change of amendment winners in two consecutive rounds~$i-1$ and $i$ 
  is that the $i$th alternative~$a_i$ on the agenda beats $b_{i-1}$.
  
  After this sanity check, 
  we can use an ILP formulation to check whether a valid guess can be realized by completing our partial profile.
  To this end, let $c_i \in\{0,1\}$, $2\le i \le y-1$, be a Boolean constant such that 
  $c_i=0$ if $a_i=b_i$; otherwise $c_i=1$.
  Further, for each two distinct alternatives~$a, b$, let $G(a,b)$ denote the set of all linear orders with $a \pref b$.
  Our ILP formulation consists of four groups of constraints: 
  \allowdisplaybreaks
  \begin{align}
    \label{cons1:2exact-completion} \sum_{\mathclap{\substack{\pref^* \in C(\pref)}}}{x(\pref, \pref^*)} & = N(\pref), && \forall \text{ partial orders } \pref,\\
    \label{cons2:bi=ai} 
    \tsumq{\pref^* \in G(b_{i-1}, a_i)}{\pref \in C^{-1}(\pref^*)} {x(\pref, \pref^*)} - \tsumlq{\pref^* \in G(a_i, b_{i-1})}{\pref \in C^{-1}(\pref^*)}{x(\pref, \pref^*)} & < c_i \cdot (n+1), && \forall 2\le i \le y-1, \\
     \label{cons3:bi-neq-ai}
    -\tsumrq{\pref^* \in G(b_{i-1}, a_i)}{\pref \in C^{-1}(\pref^*)}{x(\pref, \pref^*)} + \tsumlq{\pref^* \in G(a_i, b_{i-1})}{\pref \in C^{-1}(\pref^*)}{x(\pref, \pref^*)} & < (1-c_i)\cdot (n+1), && \forall 2\le i \le y-1, \\
    \label{cons4:p-wins}
    \tsumq{\pref^* \in G(p, a_i)}{\pref \in C^{-1}(\pref^*)}{x(\pref, \pref^*)} - \tsumlq{\pref^* \in G(a_i, p)}{\pref \in C^{-1}(\pref^*)}{x(\pref, \pref^*)} & > 0, && \forall y+1 \le i \le m \text{.}
   \end{align}
   The first group (Constraint~(\ref{cons1:2exact-completion})) ensures that each voter's partial order is completed to exactly one linear order.
   The second and the third groups of constraints make $b_i$ an \nthawinner{i}:
   If $a_i=b_i$, which means that if $c_i=0$,
   then the left-hand side of Constraint (\ref{cons2:bi=ai}) is zero.
   Hence, satisfying this constraint makes $a_i$ beat $b_{i-1}$ (that is, the number of voters preferring $a_i$ to $b_{i-1}$ is larger than the number of voters preferring $b_{i-1}$ to $a_i$).
   Otherwise, 
   $c_i=1$, the left-hand side of Constraint (\ref{cons3:bi-neq-ai}) is zero.
   Hence, satisfying this constraint makes $b_{i-1}$ beat $a_i$ (that is, the number of voters preferring $b_{i-1}$ to $a_i$ is greater than the number of voters preferring $a_i$ to $b_{i-1}$).
   The last group~(Constraint~(\ref{cons4:p-wins})) ensures that for each $i \in \{y+1,y+2,\ldots, m\}$,
   the number of voters preferring $p$ to $b_i$ is greater than the number of voters preferring $b_i $ to $p$. 
   The correctness of our ILP formulation thus follows.

   As to the running time, first of all, 
   for each valid completion of the agenda and valid sequence of the amendment winners (there are $m!\cdot m^m$ many),
   we run our ILP which has  at most $2^{m^2}\cdot m!$ variables and at most $2^{m^2}+3m$ constraints,
   and where the absolute value of each coefficient and each constant term is at most $2n$.
 \end{proof}

Using the famous result of \citet{Len83} (later improved by \citet{Kan87} and \citet{FT87}),
which states that the feasibility problem of an integer linear program
can be solved in $O(\rho_1^{2.5\rho_1+o(\rho_1)}\cdot |I|)$ time where $\rho_1$ denotes the number variables and $|I|$ the size of the integer linear program,
we can derive the following tractability result
because an integer linear program with $\rho_1$ variables and $\rho_2$ constraints,
and whose coefficients and constant terms are between $-\rho_3$ and $\rho_3$,
has $O(\rho_1\cdot \rho_2 \cdot \log{(\rho_3+2)})$ input bits.

\begin{corollary}\label{cor:possiblewinner-fpt}
  \corpossiblewinnerfpt
\end{corollary}

We close this section by remarking that the complexity result in in Corollary~\ref{cor:possiblewinner-fpt} of classification nature only.
It would be interesting to know whether 
our fixed-parameter tractability results achieved through integer linear programming 
can also be achieved by a direct combinatorial (fixed-parameter) algorithm 
(cf.~\citet[Key question 1]{BreCheFalGuoNieWoe2014}).



\subsection{Necessary Winner}
Different from the possible winner problem, notably, 
the successive and the amendment procedures have different computational complexity regarding the necessary winner problem.
Throughout the rest of this section, 
we assume that a \probNecessaryWinner instance~$(\profile=\profiletuple, p, \partialagenda)$ is given.

\paragraph{Successive procedure.}
We first consider the successive procedure and show that deciding whether our preferred \alternative is (not) necessarily a successive winner can be solved in polynomial time.
We introduce one notion regarding the alternatives in a (possibly partial) agenda.

\newcommand{\Precp}{\ensuremath{\partialagenda^\Leftarrow_p}}
\newcommand{\Succp}{\ensuremath{A^\rightarrow_p}}
\newcommand{\Preca}{\ensuremath{A^\leftarrow_a}}
\newcommand{\PSuccc}[1]{\ensuremath{\partialagenda^{\rightarrow}_{#1}}}
\newcommand{\PPrecc}[1]{\ensuremath{\partialagenda^{\leftarrow}_{#1}}}

\newcommand{\Precc}[2]{\ensuremath{{#1}^\Leftarrow_{#2}}}
\newcommand{\Succc}[2]{\ensuremath{{#1}^\Rightarrow_{#2}}}
\newcommand{\Incompc}[2]{\ensuremath{{#1}^{\sim}_{#2}}}

\begin{definition}
  Given a (possibly partial) agenda~$\partialagenda$,
  for each alternative~$c \in \alternativeset$,
  let $\Precc{\partialagenda}{c}$ be the set of all \alternatives~$c'$ that are ordered in front of $c$ by $\partialagenda$, that is, $c'\agendapref c$.
  Let $\Incompc{\partialagenda}{c}$ be the set of all \alternatives~$c'$ whose relative position to $c$ are not specified by $\partialagenda$, and 
  let $\Succc{\partialagenda}{c}$ be the set of all \alternatives~$c'$ that are ordered behind $c$ by $\partialagenda$, that is, $c\agendapref c'$.

\end{definition}

Note that for each alternative~$c\in \alternativeset$,
the three sets~$\Precc{\partialagenda}{c}$, $\Incompc{\partialagenda}{c}$, 
and $\Succc{\partialagenda}{c}$ are pairwise disjoint and that $\Precc{\partialagenda}{c} \cup \Incompc{\partialagenda}{c} \cup \Succc{\partialagenda}{c} = \alternativeset \setminus \{c\}$.

\medskip
We derive the main idea behind our polynomial-time algorithm from the following simple observation.
 
\begin{observation} \label{obs:conditions-p-not-necessary-winner}
  Alternative~$p$ is not a necessary successive winner
  if and only if there is a completion~($\profile^*, \agenda$) of the profile~$\profile$ and of the agenda~$\partialagenda$ such that some other \alternative may win, 
  that is, such that
  \begin{enumerate}
    \item\label{cond1} $p \text{ is not an } (\Succc{\agenda}{p} \cup \{p\})\text{-majority winner or} $
    \item\label{cond2} there is an \alternative~$c \in \Precc{\agenda}{p}$ that is an $(\Succc{\agenda}{c} \cup \{c\})\text{-majority winner}$.
  \end{enumerate}
\end{observation}

We will show that checking whether there is a completion satisfying one of the above conditions can be done in polynomial time.
First, we need one more notion and we need to reformulate the above conditions.

\begin{definition}
  Let $c$ be an arbitrary alternative.
  Consider a profile~$\profile^{*}$ 
  that completes the relative order of each incomparable pair~$X$ of \alternatives with $c \notin X$ in $\profile$ according to an arbitrary but fixed order.  
  We call the profile~$\profile^*$ a \emph{$c$-discriminating} profile if 
  for each alternative~$c'$ with $c'\neq c$ 
  and for each voter~$v$ that thinks $c'$ and $c$ are incomparable,
  $\profile^{*}$ completes $v$'s preference order to satisfy $c' \pref c$. 
  We call the profile~$\profile^*$ a \emph{$c$-privileging} profile if 
  for each alternative~$c'$ with $c'\neq c$ 
  and for each voter~$v$ that thinks $c'$ and $c$ are incomparable,
  $\profile^{*}$  completes $v$'s preference order to satisfy $c \pref c'$. 
\end{definition}

Note that a $c$-discriminating (resp.\ $c$-privileging) profile is unique. 
We give an example to illustrate these two concepts.

\begin{example}
  \label{ex:profile2}
  Let $\profile$ be a preference profile with four \alternatives~$a, b, c, d$,
  and two voters~$v_1, v_2$ whose partial preference orders are specified as
  \begin{align*}
    v_1\colon b \pref c \pref d, \quad v_2\colon d \pref b\text{.}
  \end{align*}
  Let the fixed order over the \alternatives be $a \pref b \pref c \pref d$.

  The profile whose voters' preference orders are 
  \begin{align*}
    v_1\colon a \pref b \pref c \pref d, \quad v_2\colon a \pref d \pref b \pref c \text{.}
  \end{align*}
  is $c$-discriminating (unique).

  The profile whose voters' preference orders are 
  \begin{align*}
    v_1\colon b \pref c \pref a \pref d, \quad v_2\colon c \pref a \pref d \pref b\text{.}
  \end{align*}
  is $c$-privileging (unique).
\end{example}

Now, we can rephrase Conditions~\ref{cond1} and \ref{cond2} (from Observation~\ref{obs:conditions-p-not-necessary-winner}) in a way that we can verify them in polynomial time.
\begin{lemma}\label{lem:claimtwo}
  \claimtwo
\end{lemma}

\begin{proof}
  For the ``only if'' case, suppose that there is a completion~$(\profile^*, \agenda)$ of $(\profile, \partialagenda)$ that satisfies condition~\ref{cond1}.
  This means that more than half of the voters in~$\profile^*$ \emph{do not} 
  prefer $p$ to $\Succc{\agenda}{p}$.
  Since $\agenda$ completes $\partialagenda$,
  this implies that $\Succc{\agenda}{p} \subseteq \alternativeset \setminus (\Precc{\partialagenda}{p} \cup \{p\})$.
  Now, consider the $p$-discriminating profile~$\profile^{**}$.
  Since each voter in $\profile^{**}$ that does not prefer $p$ to $\Succc{\agenda}{p}$ will certainly not prefer $p$ to $\Succc{\agenda}{p}$ in $\profile^{**}$ (because this profile ``discriminates'' $p$),
  it must hold that 
  more than half of the voters in $\profile^{**}$ do not prefer $p$ to $\Succc{\agenda}{p}\subseteq \alternativeset \setminus (\Precc{\partialagenda}{p} \cup \{p\})$.
  Thus, $p$ is not an $(\alternativeset \setminus \Succc{\partialagenda}{p})$-majority winner in $\profile^{**}$.

  For the ``if'' case, suppose that $p$ is not an $(\alternativeset \setminus \Precc{\partialagenda}{p})$-majority winner in the $p$-discriminating profile~$\profile^{**}$.
  Consider an arbitrary agenda~$\agenda$ that satisfies 
  \[\Precc{\partialagenda}{p}\agendapref p \agendapref (\Incompc{\partialagenda}{p} \cup \Succc{\partialagenda}{p}) \text{.}\]
  We can easily verify that $(\profile^{**}, \agenda)$ satisfies Condition~\ref{cond1} because $p$ is not a $(\Incompc{\partialagenda}{p} \cup \Succc{\partialagenda}{p})$-majority winner.
\end{proof}

\begin{lemma}\label{lem:claimthree}
  \claimthree
\end{lemma}

\begin{proof}
  For the ``only if'' case, 
  suppose that $(\profile^*, \agenda)$ is a completion of $(\profile, \partialagenda)$ which satisfies Condition~\ref{cond2} and let $c\in \Precc{\agenda}{p}$ be such an ($\Succc{\agenda}{c}\cup \{c\}$)-majority winner.
  This implies that $c$ beats $p$ in $\profile^*$ because $p \in \Succc{\agenda}{c}$.
  Observe that $c \in \Precc{\partialagenda}{p} \cup \Incompc{\partialagenda}{p}$ holds
  because $\agenda$ completes $\partialagenda$. 
  If $c\in \Incompc{\partialagenda}{p}$ holds, that is, 
  if $\partialagenda$ does not specify the relative order of $c$ and $p$,
  then the same profile~$\profile^*$ and an agenda that completes $\partialagenda$ and that satisfies $p \agendapref c$ fulfills Condition~\ref{cond1} which is not possible by our assumption.
  Thus, $c \in \Precc{\partialagenda}{p}\text{.}$
  
  Now, we show that $c$ is a $(\Succc{\partialagenda}{c}\cup \{c\})$-majority winner in the $c$-privileging completion~$\profile^{**}$ of~$\profile$.
  Observe that each voter~$v$ in $\profile^{*}$ that prefers $c$ to $\Succc{\agenda}{c}$
  must also prefer $c$ to $\Succc{\partialagenda}{c}$ because $\agenda$ completes $\partialagenda$.
  Together with the assumption that $c$ is an ($\Succc{\agenda}{c}\cup \{c\}$)-majority winner in $\profile^*$,
  this implies that more than half of the voters in $\profile^*$ prefer $c$ to $\Succc{\partialagenda}{c}$.
  Since each voter in $\profile^*$ that prefers $c$ to $\Succc{\partialagenda}{c}$ 
  will still prefer $c$ to $\Succc{\partialagenda}{c}$ in the $c$-privileging profile~$\profile^{**}$,
  more than half of the voters in $\profile^{**}$ prefer $c$ to $\Succc{\partialagenda}{c}$.

  For the ``if'' case, 
  suppose that $\Precc{\partialagenda}{p}$ contains an alternative~$b$ that is a $\Succc{\partialagenda}{b}\cup \{b\}$-majority winner in the $b$-privileging completion~$\profile^{**}$ of $\profile$.
  Now, consider an agenda~$\agenda$ that completes $\partialagenda$ and that satisfies
  \[\Precc{\partialagenda}{b} \agendapref \Incompc{\partialagenda}{b} \agendapref c \agendapref \Succc{\partialagenda}{b}\text{.}\]
  We can easily verify that $(\profile^{**},\agenda)$ satisfies Condition~\ref{cond2} (with respect to $b$).
\end{proof}

Now, we have all ingredients to show that deciding on a necessary successive winner is polynomial-time solvable.

\begin{theorem}\label{thm:necessarywinner-successive-p}
  \thmnecessarywinnersuccessive
\end{theorem}

\begin{proof}
  By Observation~\ref{obs:conditions-p-not-necessary-winner} and by Lemmas~\ref{lem:claimtwo} and \ref{lem:claimthree},
  we can conclude that 
  $p$ is not a necessary winner if and only if 
  \begin{enumerate}
    \item $p$ is not an $(\alternativeset\setminus \Precc{\partialagenda}{p})$-majority winner in the $p$-discriminating completion of $\profile$ or
    \item there is an alternative $c \in \Precc{\partialagenda}{p}$ being a $(\Succc{\partialagenda}{c}\cup \{c\})$-majority winner in the $c$-privileging completion of~$\profile$.
  \end{enumerate}
  
  \begin{algorithm}[t!]
    \SetCommentSty{\color{gray}}
    \SetAlgoVlined
    \SetKwInput{Input}{Input}
    \SetKwBlock{Block}{}{}

    \SetKw{KwNo}{``no''}
    \SetKw{KwYes}{``yes''}
 
    \Input{
        $(\profile = \profiletuple, p, \partialagenda)$ \comm{---  an instance of \probNecessaryWinner}
    }

      Compute the $p$-discriminating completion $\profile^*$ of $\profile$\\
      \If{$p$ is not an $(\alternativeset\setminus \Precc{\partialagenda}{p})$-majority winner in $\profile^*$}{
        \Return \KwNo
      }
      \ForEach{\alternative~$c \in \Precc{\partialagenda}{p}$}{\label{loop-begin}
        Compute the $c$-privileging completion $\profile^{**}$ of $\profile$\\
        \If{$c$ is a $(\Succc{\partialagenda}{c}\cup \{c\})$-majority winner in $\profile^{**}$}
        {\Return \KwNo}
      }\label{loop-end}
      \Return \KwYes
      \caption{An algorithm checking whether $p$ is a necessary successive winner.}
      \label{alg:necessary-successive}
  \end{algorithm}

  We make a remark on the second requirement:
  A $(\Succc{\partialagenda}{c}\cup \{c\})$-majority
  winner is not guaranteed to win the successive procedure because it could be the case that some other alternative in front of it in the agenda already wins the majority. 
  Nevertheless, $p$ will not be a successive winner in this case.

  We describe our general approach in Algorithm~\ref{alg:necessary-successive},
  which checks whether one of the two requirements is fulfilled.
  Fortunately, this can be done in polynomial time:
  Computing the $p$-discriminating or the~$c$-privileging completion for some alternative~$c\in \alternativeset\setminus \{p\}$ takes $O(n \cdot m^2)$ time and
  finding the majority winner also takes $O(n \cdot m^2)$ time.
  The algorithm iterates at most $m$~times through the loop in Steps~(\ref{loop-begin})-(\ref{loop-end}).
  Altogether it takes $O(n \cdot m^3)$ time.    
\end{proof}

\paragraph{Amendment procedure.}
  Adapting the \textsc{Vertex Cover} reduction from the proof of Theorem~\ref{thm:possiblewinner-amendment-nph}, 
  we can show that \probNecessaryWinner for the amendment procedure is $\coNP$-hard. 

  \begin{theorem}\label{thm:necessarywinner-amendment-nph}
    \thmnecessarywinneramendment
  \end{theorem}

  \begin{proof}
  Recall that in the proof of Theorem~\ref{thm:possiblewinner-amendment-nph}
  we constructed a profile~$\profile=(\alternativeset, V)$ with $2r-1$ voters and
  a fixed agenda~$\partialagenda$
  for a given instance  $(G=(U, E), h)$ of \textsc{Vertex Cover},   
  and we showed that
  $G$ admits a vertex cover of size at least~$(h+1)$ if and only if~$p$ is not a possible amendment winner.  
  Further, our preferred alternative~$p$ is not a possible winner if and only if the helper alternative~$b$ beats the dummy alternative~$d$ in every completion of the profile~$\profile$.
  Since all $(r-1)$ auxiliary voters and at least one vertex voter prefer each edge alternative to alternative~$b$,
  each edge alternative beats $b$.
  This implies that if $b$ beats $d$ in the second round, then $b$ will be deleted in the fourth round (note that $b$ beats $p$ in any case). 
  Since every voter has the same preference order $c_1 \pref c_2 \pref \ldots \pref c_s$ over all edge alternatives,
  edge alternative~$c_1$ beats every remaining edge alternative~$c_j$, $j\neq 1$.
  This implies that $c_1$ will necessarily win if and only if $p$ does not possibly win.
  Hence, the construction of Theorem~\ref{thm:possiblewinner-amendment-nph} provides a polynomial-time reduction from the $\coNP$-complete
  \textsc{Co-Vertex Cover} problem to our \probNecessaryWinner problem for the amendment procedure.
\end{proof}

Using the ILP formulation for \probPossibleWinner with the amendment procedure (Theorem~\ref{thm:possiblewinner-ilp}), 
we can check whether there is a possible winner different from~$p$.
Since $p$ is a necessary winner if and only there is no other possible winner,
using the results of \citet{Len83}, \citet{Kan87}, and \citet{FT87} we can conclude the following.

\begin{corollary}\label{cor:necessarywinner-amendment-fpt}
  \cornecessarywinneramendmentfpt
\end{corollary}

\subsection{The case of weighted voters}\label{sub:possible/necessarywinner-weighted}
If each \voter comes with a weight, 
then the possible winner and the necessary winner problems for the amendment procedure are already $\NP$-hard when the number of \alternatives is three and four, respectively~\citep{PinRosVenWal2011,LanPinRosSalVenWal2012}.
This is in contrast to the manipulation problem where the weighted case is
computationally ``equivalent'' to the non-weighted case: they are both
polynomial-time solvable.
We show in the following results for the successive procedure

\begin{theorem}\label{thm:wpossiblewinner-nph}
  \thmweightedpossiblewinner
\end{theorem}

\begin{proof}
  We show $\NP$-hardness by providing a polynomial-time reduction from the weakly $\NP$-complete \textsc{Partition} problem.
  Given a multi-set~$X=(x_1,x_2,$ $\dots,$ $x_r)$ of positive integers,
  \textsc{Partition} asks whether there is a \emph{perfect partition}~$X_1 \cupdot X_2 = X$ of the integers such that both parts sum up to the same value, that is, $\sum_{x \in X_1} x = \sum_{x \in X_2} x$.

  Let $X$~be a \textsc{Partition} instance with $\sum_{x \in X}x = 2K$.
  We construct a \probPossibleWinner instance~$(\profiletuple,p,\partialagenda)$ as follows.
  The set~$\alternativeset$ of alternatives consists of the preferred alternative~$p$
  and two further alternatives~$a$ and~$b$:
  \[A \coloneqq \{p, a, b\}\text{.}\]
  The set~$\voterset$ of voters consists of $r$ number voters and one dummy voter:
  \begin{enumerate}
    \item For each integer~$x_i \in X$ (which is positive), 
    we construct one \emph{number voter}~$v_i$ with the partial preference order specified by $a \pref p$ and with weight~$x_i$.
    \item We construct one additional \emph{dummy voter} with the linear preference order~$p \pref b \pref a$ and with weight~$1$.
  \end{enumerate}

  Finally, the partial agenda~$\partialagenda$ is set to $a \pref b \pref p$. 
  This completes the construction which can clearly be computed in polynomial time.

  Before going into the details of the proof, 
  observe that the total weight of all voters is $2K+1$ and, hence,
  the majority quota is $K+1$.
  
  For the ``only if'' case, assume that there is a perfect partition~$X_1 \cupdot X_2 = X$ of
  the integers such that $\sum_{x \in X_1} x = \sum_{x \in X_2} x = K$.
  Then, complete the profile as follows.
  For each number voter~$v_i$, 
  if $x_i \in X_1$, then the preference order of voter~$v_i$ is $a \pref p \pref b$;
  otherwise, the preference order of voter~$v_i$ is $b \pref a \pref p$.
  By the sequence of the agenda~$\partialagenda$,
  $a$ will be deleted because it is not a majority winner (only the voters corresponding to the integers in $X_1$ prefer $a$ to $\{p, b\}$).
  In the profile restricted to $b$ and $p$,
  the dummy voter plus the voters that correspond to the integers in $X_1$ prefer $p$ to $b$; the sum of their weights is $2K+1$.
  Thus, $p$ beats $b$ as a winner.

  For the ``if'' case, assume that there is a completion~$\profile^{*}$ of the profile such that $p$~is a successive winner.
  This implies that $a$ is not a majority winner.
  Thus, by the preference order of the dummy voter which has weight one,
  the sum of the weights of the number voters that have preference order~$b \pref a \pref p$ is at least $K$.
  In the third round, $p$ must beat $b$,
  which implies that the sum of the weights of the number voters that have preference order~$b \pref a \pref p$ is at most $K$.
  Summarizing, the number voters that have preference order~$b \pref a \pref p$ have a total weight equal to $K$.
  The corresponding integers sum up to $K$.
\end{proof}

As already mentioned in the beginning of Section~\ref{sub:possible/necessarywinner-weighted},
\citet{PinRosVenWal2011} and \citet{LanPinRosSalVenWal2012} show
that \probNecessaryWinner is weakly $\NP$-hard for the amendment procedure and for four  alternatives.
We complement this result by showing that it is polynomial-time solvable for the successive procedure,
and it is linear-time solvable for the amendment procedure when the number of alternatives is at most three.

\begin{theorem}\label{thm:w_necessary}
  \thmweightednecessarywinner
\end{theorem}

\begin{proof}
  First, we observe that the algorithm (Theorem~\ref{thm:necessarywinner-successive-p}: Algorithm~\ref{alg:necessary-successive}) that we provide to check whether an alternative is a necessary successive winner in a profile without weights can be easily adapted to also solve the case of weighted voters.
  Thus, for the successive procedure, \probWNecessaryWinner can also be solved in $O(n\cdot m^3)$ time. 
  
  Now, we focus on \probWNecessaryWinner under the amendment procedure with up to three alternatives.
  This is closely related to the problems of possible and necessary Condorcet winner. 
  To this end, for each two \alternatives $a$ and~$b$, let $w(a,b)$ be
  the sum of the weights of those voters that prefer $a$ to $b$.

  We check all (up to) six completions~$\agenda$ of the given partial
  agenda~$\partialagenda$ and decide $p$ to be a necessary amendment winner
  only if all checks give answer ``yes''.
  First, if the preferred \alternative~$p$ is at one of the first two
  positions in the agenda~$\agenda$, 
  then the problem is equivalent to asking whether
  alternative~$p$ is a necessary Condorcet winner, 
  that is, whether all completions of the preference orders can make $p$ beat
  the other(s), one by one. 
  The answer to the latter is yes if and only if 
  for each remaining \alternative~$c$,
  the sum of the weights of those voters that prefer $p$ to $c$ is more than half of the total sum of weights.
  This can be checked in linear time.
  
  We are left with the very last case where the profile has three \alternatives, denoted by $a$, $b$, and~$p$, and where the
  preferred \alternative~$p$ is at the third position in the agenda~$\agenda$. 
  In this case, $p$ is a necessary winner if and only if  
  \emph{none} of the other two \alternatives is a possible Condorcet winner which can
  be checked in linear time; denote by $K$ the total sum of the weights of all voters:
  \begin{enumerate}[(1)]
    \item\label{winner-a} If the weights of the voters preferring $a$ to $b$ sum up to more than $K/2$ and
    if the weights of the voters preferring $p$ to $a$ sum up to less than $K/2$,
    then we stop and answer ``no''.
    \item\label{winner-b} If the weights of the voters preferring $b$ to $a$ sum up to more than $K/2$, 
    and if the weights of the voters preferring $p$ to $b$ sum up to less than $K/2$,
    then we stop and answer ``no''.
    \item\label{winner-a-b} If the weights of the voters preferring $p$ to $a$ sum up to less than $K/2$ or the weights of the voters preferring $p$ to $b$ sum up to less than $K/2$,
    then we answer ``no''.
    \item Otherwise, we answer ``yes''. 
  \end{enumerate}
  We can verify that if the condition in Step (\ref{winner-a}) (resp.\ in Step (\ref{winner-b})) is satisfied, 
  then $a$ (resp.\ $b$) is a possible winner, implying that $p$ is not a necessary winner.
  If no condition in the first two steps is satisfied, 
  then there is a completion such that $a$ beats $b$ (or $b$ beats $a$).
  In this case, if one of the two conditions in the last step is satisfied,
  then the corresponding alternative can possibly win, 
  implying that $p$ is not a necessary winner. 
  Thus, the correctness follows.
\end{proof}

\newcommand{\controlmeasure}{control vulnerability ratio\xspace}
\newcommand{\manipulationmeasureA}{manipulation resistance ratio\xspace}
\newcommand{\manipulationmeasureB}{2nd winner coalition ratio\xspace}
\newcommand{\manipulationmeasureC}{smallest coalition ratio\xspace}
\newcommand{\manipulationmeasureBSize}{2nd winner coalition size\xspace}
\newcommand{\manipulationmeasureCSize}{smallest coalition size\xspace}
\section{Empirical Study}
\label{sec:experiments}

Our polynomial-time algorithms leave open how many \alternatives can win through control (or manipulation). 
We therefore investigate this empirically.
To this end,
we use the real-world data from the Preflib collection of preference profiles due to~\citet{MatWal2013} 
to investigate empirically the percentage 
of successful agenda control or manipulation.
Since only one case of the possible and the necessary winner problems is polynomial-time solvable and since Preflib offers only a very restricted variant of incomplete preferences,
we do not run experiments for these two problems.
Our results are shown in Tables~\ref{tab:exp-agenda-control} and \ref{tab:exp-manipulation}.

\subsection{Data Background}
Preflib is a library for real-world preferences.
The data is donated by various research groups.
Currently\footnote{Accessed August 20th, 2015.}, Preflib contains $314$ profiles with complete preference orders:
$100$ of them have three alternatives, $108$ of them have four alternatives, one of them has $7$~alternatives, and the remaining $105$ profiles have between $9$ and $242$ alternatives.
Among all profiles with complete preference orders, 
$135$ ones have an odd number of voters,
where $56$ of these have three alternatives,
$52$ of these have four alternatives, one of these has $7$~alternatives,
and the remaining $26$ profiles have between $14$ and $242$ alternatives.
The number of \voters ranges from $5$ to $14081$.

\begin{table}[t]
  \centering
  \begin{tabular}[t]{@{}l@{}c@{}c@{}c@{}c@{}c@{}c@{}}
    \toprule
  \multirow{2}{*}{\controlmeasure} &\phantom{cc}& \multicolumn{2}{c}{Successive} & \phantom{cc} & \multicolumn{2}{c}{Amendment}\\\cmidrule{3-4}\cmidrule{6-7}
    && $\; m\le 4\;$ & $\; m \ge 5\;$ && $\; m \le 4\;$ & $\; m \ge 5\;$\\
    \midrule
    Arithmetic mean && $0.157$ & $0.081$ && $0.000$ & $0.035$\\
    Geometric mean && $0.000$ & $0.000$ && $0.000$ & $0.000$\\
    \bottomrule
  \end{tabular}
  \caption{Experiments on agenda control with real-world data. 
     We evaluate all $135$~profiles from Preflib that have linear preference orders and that have an odd number of voters each.
     We consider profiles with $m\le 4$ and $m\ge 5$~\alternatives, respectively.
     The reason for this separation is as follows.
     First, while a large number of profiles has either three or four \alternatives (one third each),
     for $m\ge 5$, in most cases, less than \emph{five} profiles have $m$~\alternatives.
     Second, the results for profiles with up to four \alternatives are pretty different from the other profiles.
   }
  \label{tab:exp-agenda-control}
\end{table}

\subsection{Agenda Control}
For each of the $135$ profiles with odd number of voters (note that, for reasons of simplicity, we only implemented our algorithms for odd numbers of voters),
using the algorithms behind Theorems~\ref{thm:control-successive-p} and \ref{thm:control-amendment-p},
we compute the number~$m_s$ (resp.\ $m_a$) of alternatives for which a successive (resp.\ amendment) agenda control is possible.
Then, we calculate the \emph{\controlmeasure} 
as 
\begin{align*}
  \frac{m_s-1}{m-1} \text{ and } \frac{m_a-1}{m-1},
\end{align*} where $m$ denotes the number of alternatives,
respectively. Note that we have $m-1$ here since we factor out the alternative
that wins originally.
For instance, \controlmeasure~$0.5$ means that $\frac{m-1}{2}$ candidates 
are controllable.
We use both arithmetic mean and geometric mean to compute the average values among all profiles (see Table~\ref{tab:exp-agenda-control}).

Our results show that the amendment procedure tends to be more resistant than the successive procedure:
At most $3.5\%$ of the \alternatives have a chance to win the amendment procedure,
while it is $15.7\%$ for the successive procedure.

\begin{table}[t]
  \centering
  \begin{tabular}[t]{@{}l@{}c@{}c@{}c@{}c@{}c@{}c@{}}
    \toprule
  \multirow{2}{*}{Measurement} &\phantom{cc}& \multicolumn{2}{c}{Successive} & \phantom{cc} & \multicolumn{2}{c}{Amendment}\\\cmidrule{3-4}\cmidrule{6-7}
    && $\; m\le 8\;$ & $\; m \ge 9\;$ && $\; m \le 8\;$ & $\; m \ge 9\;$\\
    \midrule
    \manipulationmeasureA && $0.455$ & $0.945$ && $0.402$ & $0.924$\\
    \manipulationmeasureB && $0.288$ & $0.523$ && $0.222$ & $0.468$\\
    \manipulationmeasureC && $0.263$ & $0.386$ && $0.221$ & $0.385$\\
    \bottomrule
  \end{tabular}
  \caption{Experiments on manipulation with real-world data. 
     We evaluate all $314$~profiles from Preflib that have linear preference orders.
     We consider profiles with $m\le 8$ and $m\ge 9$~\alternatives, respectively.
     The reason for this separation is that
     we only consider all $m!$ possible agendas when $m\le 8$.
     We use geometric mean to compute the average.
   }
  \label{tab:exp-manipulation}
\end{table}

\subsection{Manipulation}
Since Preflib does not offer any agenda, we have to generate a set~$X$ of agendas for 
manipulation 
to obtain a good representation.
The size of $X$ depends on the number~$m$ of \alternatives:
\begin{enumerate} 
\item 
If $m \le 8$, then we let $X$ be the set of all possible agendas, that is, $|X|\coloneqq m!$.
\item 
Otherwise, we generate a set~$X$ of uniformly distributed random agendas with $|X| \coloneqq \min(n^2, 8!)$, where $n$ denotes the number of voters in the input.
\end{enumerate} 
Then, for each \alternative~$c$ and for each agenda~$\agenda \in X$, 
using the algorithms behind Theorems~\ref{thm:manipulation-successive-p} and \ref{thm:manipulation-amendment-p},
we compute 
the minimum coalition size, that is, the minimum number of voters needed to make $c$ a successive (resp. an amendment) winner. 
Let this number 
for the successive (resp.\ amendment) procedure 
be 
$\kappa_s(\profile,c,\agenda)$ (resp.\ $\kappa_a(\profile,c,\agenda)$). 
This is upper-bounded by~$n+1$.
Then, we calculate the \emph{manipulation resistance ratio} as 
\begin{align*}
  \frac{\sum\limits_{\agenda\in X} \sum\limits_{c\in C} \kappa_s(\profile,c,\agenda)}{|X|\cdot (m-1) \cdot (n+1)} \text{ and } \frac{\sum\limits_{\agenda\in X} \sum\limits_{c\in C} \kappa_a(\profile,c,\agenda)}{|X|\cdot (m-1) \cdot (n+1)}\text{.}
\end{align*}

Since most \alternatives need a coalition of more than $n$~voters to
manipulate successfully,
which strongly affects the manipulation resistance ratio (see the first row in Table~\ref{tab:exp-manipulation}),
we also consider two related concepts: 
\begin{itemize}
  \item The ratio of the \emph{\manipulationmeasureBSize}, 
that is, 
the coalition size for the \alternative that becomes a winner after the original winner is removed.
Formally, the \manipulationmeasureB is defined as
\begin{align*}
  \frac{\sum\limits_{\agenda\in X} \kappa_s(\profile,c^{*},\agenda)}{|X|\cdot(n+1)} \text{ and } \frac{\sum\limits_{\agenda\in X} 
    \kappa_a(\profile,c^{*},\agenda)}{|X|\cdot(n+1)}\text{,}
\end{align*}
where $c^*$ is a successive (resp.\ an amendment) winner after the original winner is removed.
\item The ratio of the \emph{smallest coalition size},
that is,
the size~$\sigma_s(\profile,\agenda)$ (resp.\ $\sigma_a(\profile,\agenda)$) of the smallest coalition that makes any \alternative win.
Formally, the \manipulationmeasureC is defined as
\begin{align*}
  \frac{\sum\limits_{\agenda\in X} \sigma_s(\profile,\agenda)}{|X|\cdot(n+1)} \text{ and } \frac{\sum\limits_{\agenda\in X} \sigma_a(\profile,\agenda)}{|X|\cdot(n+1)}\text{.}
\end{align*}
\end{itemize}
Our results show that successful manipulations with few voters are rare:
For profiles with up to eight \alternatives 
the average coalition size is at least $0.4\cdot(n+1)$
(even the $2$nd winner coalition size is $0.2\cdot(n+1)$;
the smallest coalition size is only slightly lower),
while for profiles with at least five \alternatives 
the average coalition size is almost $n+1$
(even the $2$nd winner coalition size is roughly $0.5\cdot(n+1)$).


\section{Conclusion}
\label{sec:conclusion}
Our work indicates that, from a computational 
perspective, 
the amendment procedure seems 
superior to the 
successive procedure against agenda control and strategic voting. 
Our work supports the claim
that most European and Latin parliaments (cf.~\citet{ApeBalMas2014}) should 
rather go the Anglo-American way, that is, they should 
use amendment procedures instead of successive procedures.

Following the spirit of \citet{BetHemNie2009}, 
it would be of interest to complement our computational
hardness results for possible and necessary winner problems 
with a refined complexity analysis concerning tractable special cases.
For instance, our $\NP$-hardness reductions for the possible winner problems assume that voters may have arbitrary partial preferences.
It would be interesting to study whether this still holds if voter preferences are single-peaked or semi-single-peaked~\citep{Black1958}.

For all our intractable problems,
we obtain fixed-parameter tractability with respect to the parameter ``number~$m$ of alternatives''.
With respect to the parameter ``number~$n$ of voters'',
however, we could not settle their parameterized complexity. 

It would be natural to also adopt a more game-theoretic 
view on the strategic voting scenarios we considered.
Finally, it would be also interesting to study further 
manipulation scenarios for parliamentary voting procedures including,
for example, candidate control 
as discussed by \citet{Ras2014}.

\paragraph{Acknowledgements}
Robert Bredereck was supported by the German Research Foundation (DFG), research project PAWS, NI~369/10.
Toby Walsh was supported by the Alexander von Humboldt Foundation while visiting TU Berlin.
NICTA  is  funded  by  the  Australian  Government  through
the Department of Communications and the Australian Research Council
through  the  ICT  Centre  of  Excellence  Program.



\begin{thebibliography}{37}
\providecommand{\natexlab}[1]{#1}
\providecommand{\url}[1]{\texttt{#1}}
\expandafter\ifx\csname urlstyle\endcsname\relax
  \providecommand{\doi}[1]{doi: #1}\else
  \providecommand{\doi}{doi: \begingroup \urlstyle{rm}\Url}\fi

\bibitem[Apesteguia et~al.(2014)Apesteguia, Ballester, and
  Masatlioglu]{ApeBalMas2014}
J.~Apesteguia, M.~A. Ballester, and Y.~Masatlioglu.
\newblock A foundation for strategic agenda voting.
\newblock \emph{Games and Economic Behavior}, 87:\penalty0 91--99, 2014.

\bibitem[Aziz et~al.(2012)Aziz, Harrenstein, Brill, Lang, Fischer, and
  Seedig]{AziHarBriLanSee2012}
H.~Aziz, P.~Harrenstein, M.~Brill, J.~Lang, F.~Fischer, and H.~G. Seedig.
\newblock Possible and necessary winners of partial tournaments.
\newblock In \emph{Proceedings of the 11th International Conference on
  Autonomous Agents and Multiagent Systems (AAMAS~'12)}, pages 585--592, 2012.

\bibitem[Banks(1985)]{Ban1985}
J.~S. Banks.
\newblock Sophisticated voting outcomes and agenda control.
\newblock \emph{Social Choice and Welfare}, 1\penalty0 (4):\penalty0 29--306,
  1985.

\bibitem[Barber\`{a} and Gerber(2014)]{BarGer2014}
S.~Barber\`{a} and A.~Gerber.
\newblock Sequential voting and agenda manipulation.
\newblock Technical report, Universitat Aut{\`o}noma de Barcelona and Barcelona
  GSE, 2014.

\bibitem[{Bartholdi III} and Orlin(1991)]{BarOrl1991}
J.~J. {Bartholdi III} and J.~B. Orlin.
\newblock Single transferable vote resists strategic voting.
\newblock \emph{Social Choice and Welfare}, 8\penalty0 (4):\penalty0 341--354,
  1991.

\bibitem[{Bartholdi III} et~al.(1989){Bartholdi III}, Tovey, and
  Trick]{BarTovTri1989}
J.~J. {Bartholdi III}, C.~A. Tovey, and M.~A. Trick.
\newblock The computational difficulty of manipulating an election.
\newblock \emph{Social Choice and Welfare}, 6\penalty0 (3):\penalty0 227--241,
  1989.

\bibitem[Betzler et~al.(2009)Betzler, Hemmann, and Niedermeier]{BetHemNie2009}
N.~Betzler, S.~Hemmann, and R.~Niedermeier.
\newblock A multivariate complexity analysis of determining possible winners
  given incomplete votes.
\newblock In \emph{Proceedings of the 21st International Joint Conference on
  Artificial Intelligence (IJCAI~'09)}, pages 53--58, 2009.

\bibitem[Black(1958)]{Black1958}
D.~Black.
\newblock \emph{The Theory of Committees and Elections}.
\newblock Cambridge University Press, 1958.

\bibitem[Bredereck et~al.(2014)Bredereck, Chen, Faliszewski, Guo, Niedermeier,
  and Woeginger]{BreCheFalGuoNieWoe2014}
R.~Bredereck, J.~Chen, P.~Faliszewski, J.~Guo, R.~Niedermeier, and G.~J.
  Woeginger.
\newblock Parameterized algorithmics for computational social choice: Nine
  research challenges.
\newblock \emph{Tsinghua Science and Technology}, 19\penalty0 (4):\penalty0
  358--373, 2014.

\bibitem[Bredereck et~al.(2015)Bredereck, Chen, Niedermeier, and
  Walsh]{BreCheNieWal2015}
R.~Bredereck, J.~Chen, R.~Niedermeier, and T.~Walsh.
\newblock Parliamentary voting procedures: {A}genda control, manipulation, and
  uncertainty.
\newblock In \emph{Proceedings of the 24rd International Joint Conference on
  Artificial Intelligence (IJCAI~'15)}, pages 164--170. AAAI Press, 2015.

\bibitem[Conitzer et~al.(2007)Conitzer, Sandholm, and Lang]{ConSanLanjacm2007}
V.~Conitzer, T.~Sandholm, and J.~Lang.
\newblock When are elections with few candidates hard to manipulate?
\newblock \emph{Journal of the ACM}, 54\penalty0 (3):\penalty0 1--33, 2007.

\bibitem[Cygan et~al.(2015)Cygan, Fomin, Kowalik, Lokshtanov, Marx, Pilipczuk,
  Pilipczuk, and Saurabh]{CyFoKoLoMaPiPiSa2015}
M.~Cygan, F.~V. Fomin, L.~Kowalik, D.~Lokshtanov, D.~Marx, M.~Pilipczuk,
  M.~Pilipczuk, and S.~Saurabh.
\newblock \emph{Parameterized Algorithms}.
\newblock Springer, 2015.

\bibitem[Downey and Fellows(2013)]{DF13}
R.~G. Downey and M.~R. Fellows.
\newblock \emph{Fundamentals of Parameterized Complexity}.
\newblock Springer, 2013.

\bibitem[Enelow and Koehler(1980)]{EneKoe1980}
J.~M. Enelow and D.~H. Koehler.
\newblock The amendment in legislative strategy: {S}ophisticated voting in the
  {U.S.} {C}ongress.
\newblock \emph{The Journal of Politics}, 42\penalty0 (2):\penalty0 396--413,
  1980.

\bibitem[Farquharson(1969)]{Far1969}
F.~Farquharson.
\newblock \emph{Theory of Voting}.
\newblock Yale University Press, 1969.

\bibitem[Flum and Grohe(2006)]{FG06}
J.~Flum and M.~Grohe.
\newblock \emph{Parameterized Complexity Theory}.
\newblock Springer, 2006.

\bibitem[Frank and Tardos(1987)]{FT87}
A.~Frank and {\'E}.~Tardos.
\newblock An application of simultaneous {D}iophantine approximation in
  combinatorial optimization.
\newblock \emph{Combinatorica}, 7\penalty0 (1):\penalty0 49--65, 1987.

\bibitem[Harary and Moser(1966)]{HarMos1966}
F.~Harary and L.~Moser.
\newblock The theory of round robin tournaments.
\newblock \emph{The American Mathematical Monthly}, 73\penalty0 (3):\penalty0
  231--246, 1966.

\bibitem[Hazon et~al.(2012)Hazon, Aumann, Kraus, and
  Wooldridge]{HazAumKraWoo2012}
N.~Hazon, Y.~Aumann, S.~Kraus, and M.~Wooldridge.
\newblock On the evaluation of election outcomes under uncertainty.
\newblock \emph{Artificial Intelligence}, 189:\penalty0 1--18, 2012.

\bibitem[Jung(1989)]{Jun1989}
J.~P. Jung.
\newblock Condorcet consistent binary agendas under incomplete information.
\newblock In \emph{Models of Strategic Choice in Politics}. University of
  Michigan Press, 1989.

\bibitem[Kannan(1987)]{Kan87}
R.~Kannan.
\newblock Minkowski's convex body theorem and integer programming.
\newblock \emph{Mathematics of Operations Research}, 12\penalty0 (3):\penalty0
  415--440, 1987.

\bibitem[Konczak and Lang(2005)]{KonLan2005}
K.~Konczak and J.~Lang.
\newblock Voting procedures with incomplete preferences.
\newblock In \emph{Proceedings of the 19st International Joint Conference on
  Artificial Intelligence (IJCAI~'05)}, pages 124--129, 2005.

\bibitem[Lang et~al.(2012)Lang, Pini, Rossi, Salvagnin, Venable, and
  Walsh]{LanPinRosSalVenWal2012}
J.~Lang, M.~S. Pini, F.~Rossi, D.~Salvagnin, K.~B. Venable, and T.~Walsh.
\newblock Winner determination in voting trees with incomplete preferences and
  weighted votes.
\newblock \emph{Autonomous Agents and Multi-Agent Systems}, 25\penalty0
  (1):\penalty0 130--157, 2012.

\bibitem[Lenstra(1983)]{Len83}
H.~W. Lenstra.
\newblock Integer programming with a fixed number of variables.
\newblock \emph{Mathematics of Operations Research}, 8\penalty0 (4):\penalty0
  538--548, 1983.

\bibitem[Mattei and Walsh(2013)]{MatWal2013}
N.~Mattei and T.~Walsh.
\newblock Preflib: {A} library for preferences http://www.preflib.org.
\newblock In \emph{Proceedings of the 3rd International Conference on
  Algorithmic Decision Theory (ADT~'13)}, volume 8176 of \emph{Lecture Notes in
  Computer Science}, pages 259--270, 2013.
\newblock http://preflib.org.

\bibitem[McKelvey and Niemi(1978)]{MN81}
R.~D. McKelvey and R.~G. Niemi.
\newblock A multistage game representation of sophisticated voting for binary
  procedures.
\newblock \emph{Journal of Economic Theory}, 18\penalty0 (1):\penalty0 1--22,
  1978.

\bibitem[Miller(1977)]{Mil1977}
N.~R. Miller.
\newblock Graph-theoretical approaches to the theory of voting.
\newblock \emph{American Journal of Political Science}, 21\penalty0
  (4):\penalty0 769--803, 1977.

\bibitem[Moulin(1986)]{Mou86}
H.~Moulin.
\newblock Choosing from a tournament.
\newblock \emph{Social Choice and Welfare}, 3\penalty0 (4):\penalty0 271--291,
  1986.

\bibitem[Niedermeier(2006)]{Nie06}
R.~Niedermeier.
\newblock \emph{Invitation to Fixed-Parameter Algorithms}.
\newblock Oxford University Press, 2006.

\bibitem[Ordeshook and Palfrey(1988)]{OrdPal1988}
P.~C. Ordeshook and T.~R. Palfrey.
\newblock Agendas, strategic voting, and signaling with incomplete information.
\newblock \emph{American Journal of Political Science}, 32\penalty0
  (2):\penalty0 441--466, 1988.

\bibitem[Ordeshook and Schwartz(1987)]{OrdSch1987}
P.~C. Ordeshook and T.~Schwartz.
\newblock Agendas and the control of political outcomes.
\newblock \emph{The American Political Science Review}, 81\penalty0
  (1):\penalty0 179--200, 1987.

\bibitem[Pini et~al.(2011)Pini, Rossi, Venable, and Walsh]{PinRosVenWal2011}
M.~S. Pini, F.~Rossi, K.~B. Venable, and T.~Walsh.
\newblock Incompleteness and incomparability in preference aggregation:
  Complexity results.
\newblock \emph{Artificial Intelligence}, 175\penalty0 (7-8):\penalty0
  1272--1289, 2011.

\bibitem[Rasch(2000)]{Ras2000}
B.~E. Rasch.
\newblock Parliamentary floor voting procedures and agenda setting in {E}urope.
\newblock \emph{Legislative Studies Quarterly}, 25\penalty0 (1):\penalty0
  3--23, 2000.

\bibitem[Rasch(2014)]{Ras2014}
B.~E. Rasch.
\newblock Insincere voting under the successive procedure.
\newblock \emph{Public Choice}, 158\penalty0 (3--4):\penalty0 499--511, 2014.

\bibitem[Shepsle and Weingast(1984)]{SW84}
K.~A. Shepsle and B.~R. Weingast.
\newblock A multistage game representation of sophisticated voting for binary
  procedures.
\newblock \emph{American Journal of Political Science}, 28\penalty0
  (1):\penalty0 49--74, 1984.

\bibitem[Walsh(2007)]{Wal2007}
T.~Walsh.
\newblock Uncertainty in preference elicitation and aggregation.
\newblock In \emph{Proceedings of the 22nd Conference on Artificial
  Intelligence (AAAI~'07)}, pages 3--8. AAAI Press, 2007.

\bibitem[Xia and Conitzer(2011)]{XiaCon2011}
L.~Xia and V.~Conitzer.
\newblock Determining possible and necessary winners under common voting rules
  given partial orders.
\newblock \emph{Journal of Artificial Intelligence Research}, 41:\penalty0
  25--67, 2011.

\end{thebibliography}

\end{document}